\documentclass[12pt]{article}

\usepackage{subfigure}

 \usepackage[utf8]{inputenc}
  \usepackage[noTeX]{mmap}
  \usepackage[T1]{fontenc}
\usepackage{authblk}

\usepackage{latexsym}
\usepackage{graphics}
\usepackage{amsmath}
\usepackage{xspace}
\usepackage{amssymb}
\usepackage{psfrag}
\usepackage{epsfig}
\usepackage{amsthm}
\usepackage{pst-all}

\usepackage{color}

\definecolor{Red}{rgb}{1,0,0}
\definecolor{Blue}{rgb}{0,0,1}
\definecolor{Olive}{rgb}{0.41,0.55,0.13}
\definecolor{Yarok}{rgb}{0,0.5,0}
\definecolor{Green}{rgb}{0,1,0}
\definecolor{MGreen}{rgb}{0,0.8,0}
\definecolor{DGreen}{rgb}{0,0.55,0}
\definecolor{Yellow}{rgb}{1,1,0}
\definecolor{Cyan}{rgb}{0,1,1}
\definecolor{Magenta}{rgb}{1,0,1}
\definecolor{Orange}{rgb}{1,.5,0}
\definecolor{Violet}{rgb}{.5,0,.5}
\definecolor{Purple}{rgb}{.75,0,.25}
\definecolor{Brown}{rgb}{.75,.5,.25}
\definecolor{Grey}{rgb}{.5,.5,.5}

\newcommand{\G}{\mathbb{G}}

\newcommand{\ignore}[1]{\relax}

\newcommand{\ER}{Erd{\"o}s-R\'{e}nyi }

\definecolor{Red}{rgb}{1,0,0}
\definecolor{Blue}{rgb}{0,0,1}
\definecolor{Olive}{rgb}{0.41,0.55,0.13}
\definecolor{Green}{rgb}{0,1,0}
\definecolor{MGreen}{rgb}{0,0.8,0}
\definecolor{DGreen}{rgb}{0,0.55,0}
\definecolor{Yellow}{rgb}{1,1,0}
\definecolor{Cyan}{rgb}{0,1,1}
\definecolor{Magenta}{rgb}{1,0,1}
\definecolor{Orange}{rgb}{1,.5,0}
\definecolor{Violet}{rgb}{.5,0,.5}
\definecolor{Purple}{rgb}{.75,0,.25}
\definecolor{Brown}{rgb}{.75,.5,.25}
\definecolor{Grey}{rgb}{.5,.5,.5}
\definecolor{Pink}{rgb}{1,0,1}
\definecolor{DBrown}{rgb}{.5,.34,.16}
\definecolor{Black}{rgb}{0,0,0}

\usepackage{bbm}
\usepackage{mathrsfs}
\usepackage{epsfig}
\usepackage{soul,color}
\usepackage{graphicx}
\usepackage{amsfonts}
\usepackage{amsthm}
\usepackage[figuresright]{rotating}
\usepackage{amssymb}
\usepackage{amsmath}
\usepackage{dcolumn}
\usepackage{physics}
\usepackage{float}
\usepackage{bm}
\usepackage{verbatim}
\usepackage[normalem]{ulem}
\usepackage[colorlinks,linkcolor=blue,anchorcolor=blue,citecolor=blue,urlcolor=blue]{hyperref}
\usepackage{mathtools}

\usepackage{algorithm}
\usepackage{algorithmic}

\usepackage{xcolor}
\newtheorem{theorem}{Theorem}
\newtheorem{lemma}[theorem]{Lemma}

\newtheorem{proposition}[theorem]{Proposition}

\newcommand{\lr}[1]{\left(#1\right)}

\newcommand{\E}{\mathbb{E}}

\newcommand{\C}{\mathbb{C}}

\newcommand{\R}{\mathbb{R}}
\newcommand{\Z}{\mathbb{Z}}

\newcommand{\cE}{\mathcal{E}}
\newcommand{\cF}{\mathcal{F}}

\newcommand{\cH}{\mathcal{H}}

\newcommand{\cR}{\mathcal{R}}

\newcommand{\cT}{\mathcal{T}}

\newcommand{\id}{\mathbb{I}}

\begin{document}

\title{The free energy limit of the SYK model at high temperature}

\author{
David Gamarnik \thanks{Operations Research Center and Center for Statistics and Data Sciences, IDSS, Massachusetts Institute of Technology,
\href{mailto:gamarnik@mit.edu}{gamarnik@mit.edu}},
Francisco Pernice \thanks{EECS, Massachusetts Institute of Technology, \href{mailto:fpernice@mit.edu}{fpernice@mit.edu}},
Alexander Schmidhuber  \thanks{Center for Theoretical Physics, Massachusetts Institute of Technology, \href{mailto:schmidhuber.alexander@gmail.com}{alexsc@mit.edu}},
Alexander Zlokapa  \thanks{Center for Theoretical Physics, Massachusetts Institute of Technology, \href{mailto:azlokapa@mit.edu}{azlokapa@mit.edu}}}

\maketitle

\begin{abstract}
The Sachdev-Ye-Kitaev (SYK) model is a disordered quantum mean-field model studied in condensed matter physics and the holographic theory of black holes. Its structural properties  can be derived heuristically using a combination of the replica method and path integration techniques. Analyzing it mathematically rigorously, however, turned out to be notoriously difficult, even for  basic questions such as computing the annealed free energy.  

In this paper we rigorously compute the free energy limit (annealed and quenched) for this model at high enough but constant temperature. 
Our results are in numerical agreement with the results derived by physics methods. Remarkably, though, our method of proof is novel and is different from the physics approach. It is based on  (a) the theory of the component structure of sparse random graphs and 
(b) a variant of the cavity method, used widely  in prior rigorous and heuristic treatments of classical spin glasses.
\end{abstract}

% \newpage
\tableofcontents
% \newpage

\section{Introduction}
The Sachdev-Ye-Kitaev (SYK) model~\cite{chowdhury2022sachdev} is a disordered quantum mean-field model defined in terms of a random quantum Hermitian Hamiltonian operator $H_{\rm SYK}$ in the Hilbert 
space $\C^{2^{n\over 2}}$.
The  definition of the model is found at the beginning of the next section. 
Introduced originally by Sachdev and Ye~\cite{sachdev1993gapless} in the context of condensed matter physics, its relevance to the physics of black holes was recognized
by Kitaev~\cite{kitaev2015simple} in connection with the so-called holographic duality theory. 
The model is well-studied in the physics literature,
where most basic properties of the system were  derived by employing  standard physics methods, including the path integration and the replica methods.
Initial analysis of the system was conducted by Maldacena and Stanford~\cite{maldacena2016remarks} 
who obtained various macroscopic properties of the system including quenched
and annealed free energy (which happen to be the same for this model, see below), out-of-time-order correlators (OTOCs), and many other properties. 
This was further developed in many subsequent works in physics, too many to list, see~\cite{chowdhury2022sachdev} for a survey; we also note that a recent paper derives similar properties as SYK but from a model based on a non-random Hamiltonian~\cite{biggs2026melonic}.

Mathematically rigorous literature on SYK, however, is more limited as the model proved to be notoriously difficult to analyze, in large part due to its quantum nature.
We now describe some of the recent works, mostly focusing on the ones relevant to the questions raised in our paper. A fairly straightforward argument was used
in Feng et al~\cite{feng2019spectrum}
to show that the ground state energy of $H_{\rm SYK}$ is with high probability (w.h.p.) at most  $n\sqrt{\log 2}$ (using our choice of normalization). 
For the special case of $q=2$ body interaction, the ground state value can be computed using random matrix theory, as shown in~\cite{feng2019spectrum} as well. 
In the same
paper the authors compute the spectrum of the Hamiltonian and establish that it behaves according to the Gaussian law in some sense. An important development
is by Hastings and O'Donnell~\cite{hastings2022optimizing} who provided a matching (up to a constant) lower bound $\Theta(n)$ 
on the ground state energy, which is algorithmically
certifiable and runs in polynomial time. Namely, they have constructed a polynomial time algorithm which produces a quantum state with energy at least $cn$ 
for some constant $c>0$ whenever it exists, which it does with w.h.p. 
Chen and Lucas~\cite{chen2021operator} and  Lucas~\cite{lucas2020non} 
studied the scrambling times of the $H_{SYK}$, namely the growth rate
of the operator $[A(t),B]$ where $A(t)=e^{it H_{\rm SYK}}Ae^{-it H_{\rm SYK}}$ and $A$ and $B$ are commuting
local operators. Consistently with the physics predictions, the latter work shows that  
it takes $t=\Theta(\log n)$ time for the operator to reach order $O(1)$ in the appropriate norm. Herasymenko et al~\cite{herasymenko2023optimizing} 
show that so-called Gaussian states (which include product states)
cannot reach energy order $\Theta(n)$ in the SYK model, and this in particular elucidates why constructing near ground states algorithmically appears hard. 
Further evidence of the potential algorithmic hardness of the problem of constructing near ground states is given in Anschuetz et al~\cite{anschuetz2025strongly} 
which established
that the minimal circuit size for building a near ground state is at least $n^{\Theta(q/2)}$. Another important result in the same paper shows that the quenched and annealed free energy limits coincide. Namely, for every $\beta\in\R_+,$
\begin{align}
{1\over \beta n}\E \log \Tr\left(e^{\beta H_{\rm SYK}}\right)={1\over \beta n}\log\E\Tr\left(e^{\beta H_{\rm SYK}}\right)+o(1), \label{eq:quenched-is-annealed}
\end{align}
Hence, the model exhibits a ``non-glassy'' property, unlike analogous disordered models that replace fermions with Paulis or Ising spins. 
Notably, (\ref{eq:quenched-is-annealed}) does not imply the existence of either of the limits, which by itself is surprisingly a non-trivial property
to establish, as we discuss below. 
The asymptotic equivalence of dense and sparse
SYK models (a form of universality) was established in Anschuetz et al \cite{anschuetz2025bounds} using Lindeberg's approach. 
We take advantage of this result here when we switch from the Gaussian to the Rademacher disorder in order to simplify the proofs. 
An algorithmic question of estimating local observables
was considered in Kiani and Zlokapa~\cite{zlokapa2026syk}. Using physics methods it was shown how estimates can be computed efficiently in polynomial time.  
A rigorous counterpart of these results in the high temperature regime was obtained very recently by Zlokapa~\cite{zlokapa2026rigorous}. As of now
it remains open if the SYK model is amenable to fast algorithms in the low temperature (large $\beta$) regime.

Despite this recent progress, even some of the basic questions regarding the SYK model remain unanswered, including the value of the ground state energy 
limit, the limits of the annealed and quenched free energies, and even the existence of these limits, from the mathematically rigorous point of view.
In particular, can these values, derived using physics methods, be validated by mathematically rigorous techniques? It is notable that some of the classical counterparts
of these quantities are trivial to compute. For example the annealed free energy of the classical spin glass model with $n$-spins  is just 
$n\log 2/\beta+(n/2)\beta$, obtained trivially as the logarithm of the expectation of 
the sum of $2^n$ many exponents of centered normal random variables. Similarly, for many models
on sparse random graphs, the annealed free energy corresponds to the so-called first moment method and 
is straightforward to compute due to linearity of expectations~\cite{alon2004probabilistic}. Yet computing the annealed free energy
for quantum random Hamiltonians, such as the SYK model's, has turned out to be far more difficult. 

Computing the  quenched free energy and the ground state energies  is well-recognized to be a challenging problem, however, even in classical spin glasses.
In particular, in the physics literature this was done first by Parisi~\cite{parisi1979infinite,parisi1980sequence}, 
prompting him to introduce the celebrated Replica Symmetry Breaking (RSB) method. Mathematically rigorous backing of the replica-based predictions took place
much later. First Guerra and Toninelli~\cite{GuerraTon} proved the existence of the quenched free energy limit using a Gaussian comparison inequality type argument. 
As  mentioned earlier, even proving the existence of such limits was a mathematically non-trivial task. Then Guerra in~\cite{guerra2003broken}
established that the heuristic answer for the quenched free energy derived by Parisi using the RSB method is a rigorous upper bound on the
quenched free energy. The final step was obtained in a breakthrough work by Talagrand~\cite{talagrand2006parisi} who proved a matching lower bound.
See~\cite{TalagrandBook,panchenko2013Sherrington} for a book treatment of the subject. It is remarkable that neither the upper nor the lower
bound techniques were based on the RSB technique, which to the day remains a mathematically non-rigorous though highly accurate physics heuristic,
as is the path integration method used widely in the quantum physics literature.

\subsubsection*{The main result and the proof technique}
In this paper we compute the annealed (and therefore by~\cite{anschuetz2025strongly} also the quenched) free energy  
limit of the SYK model at small enough but constant inverse temperature parameter $\beta$. In particular, we prove that this limit exists in the same 
range of $\beta$.
The answer is provided in terms of the solution of a certain functional fixed point
equation and can be easily computed numerically. As we report in Section~\ref{section:Numerics}, our results are in perfect agreement with the results
based on physics methods in~\cite{maldacena2016remarks} for the range of $\beta$ we managed to conduct the computations, which is $\beta \in [0,3]$.
Our result also establishes a lower bound $\Theta(n)$ on the ground state energy limit, similarly to~\cite{hastings2022optimizing}. 
This is done by verifying numerically that the free energy limit has a positive value at positive small $\beta$  and then using a straightforward bound
between the free energy and the ground state energy values.

Importantly, our proof technique departs from the path integration and replica techniques, and is based instead on two entirely different ingredients, which 
are (a) theory of the component structure of sparse random graphs and (b) the cavity method. The analytic form of our answer also differs from that 
of~\cite{maldacena2016remarks}, and the agreement between the two answers is only validated numerically. We believe though that the two answers
can be reconciled analytically and leave it as an interesting open question for further research. 
We now elaborate on the techniques (a) and (b) above, and then explain
how they feature in our approach. The component structure of sparse random graphs was one of the earliest developments in the theory of random
graphs~\cite{BollobasBook,alon2004probabilistic,janson2011random} and can be summarized roughly as follows. Consider a sparse
random \ER graph $\G(n,d/n)$ which is obtained by pairing every two of $n$ nodes with probability $d/n$ independently across all pairs. 
Here $d$ is independent of $n$.
Then when $d<1$, the connected components of this graph are $O(1)$ in size on average, and even the largest component has size  only $\Theta(\log n)$ w.h.p.
Conversely, when $d>1$, the largest component is w.h.p. $\Theta(n)$,  often dubbed as the ''giant'' component. The component structure in the $d<1$ regime
is further described  as a simple subcritical branching processes with Poisson out-degree distribution with parameter $d$. 

The cavity method was pioneered in the physics of spin glasses by  Mezard, Parisi and Virasoro~\cite{MezardParisiVirasoro}. It is a remarkable
exception from  other physics methods in the sense that it admits a mathematically rigorous treatment, at least for some models and in certain regimes.
Letting $Z_n$ denote the partition function of some $n$-spin model, the method is based on representing the free energy $\log Z_n$ 
as a telescoping sum $\sum_{k\le n} (\log Z_k-\log Z_{k-1})$ and then representing the change $\log (Z_k/Z_{k-1})$
(a combinatorial derivative of a kind) in terms of the marginal distribution of the $k$-th spin added to the $k-1$-spin size model. To the extent
that such a distribution can be computed, the method leads to  the estimation of the free energy and ground state values. 
One of the earliest
mathematically rigorous treatments of the cavity method (though not called so back then) is Aldous' work on random assignment 
 mean field model~\cite{aldous1992asymptotics,aldous2001zeta,AldousSteele:survey}. 
As a mathematical method, the cavity trick turned out to be particularly powerful
in the case of sparse models such as $\G(n,d/n)$ or its random regular counterpart, as demonstrated 
in Gamarnik et al~\cite{gamarnik2006maximum,gamarnik2008invariant}, (see 
also~\cite{AldousSteele:survey,GamarnikGoldbergRegGraphs,gamarnik2010ptas,gamarnik2013correlation,GamarnikTutORialsInforms}).
The cavity method proved to be effective  even for non-random (worst-case) models as an approach for computing  partition 
functions~\cite{BandyopadhyayGamarnikCounting,weitzCounting}. Applicability of the cavity method to lattices was also demonstrated
in~\cite{GamarnikKatzSequentialCavity} and dubbed  the Sequential Cavity Method.  
As we remark in the last section, we believe that the Sequential Cavity  Method might prove to be useful for quantum lattice models as well. 
As mentioned earlier, the cavity method is effective so long as the 
marginal distribution of the spins can be computed, and the latter is typically done by  establishing
the correlation decay property.  This effectively reduces the computations on a big graph  to the one on small order $O(1)$ subgraph, usually a tree. 
The computation is done
by conducting the belief propagation type iteration on marginal distributions. A recent application of the cavity method to the dynamics of spin glasses 
was obtained by the first two authors in~\cite{dandi2026rigorous}.

We now  summarize our approach and concurrently explain how the theory of the component structure of random graphs 
and the cavity method are employed as the main ingredients.
We start by expanding the partition function $\E[\Trace e^{\beta H_{\rm SYK}}]$ as 
\begin{align*}
\sum {\beta^m\over m!} \Tr\E H_{\rm SYK}^m=\sum {\beta^m\over m!}\sum_{\bar i_1,\ldots,\bar i_m} \Tr\E \prod H_{\bar i},
\end{align*}
where $H_{\bar i}$ are local Majorana terms in the Hamiltonian $H_{\rm SYK}$ and $\Tr$ is the trace operator. 
It is not hard to show (though technically not required for our derivation) that the
dominating contribution to the sum above comes from the terms when  $m$ is order $n$. In this case a typical 
element $\bar i_1,\ldots,\bar i_m$ of the sum above can be viewed as a sparse
random hypergraph with $m$ hyperedges induced by $q$-tuples $\bar i_1,\ldots, \bar i_m$. 
It is also possible to show that the dominant contribution comes from the case when each hyperedge appears
exactly twice, thus leading to the well-recognized chordal structure of the Majorana terms, used widely in heuristic physics derivations (e.g.,~\cite{garcia2018exact}). This also implies that the product $\prod H_{\bar i}$ is always the identity operator up to sign, and determining
the sign accurately is the ''name of the game'' for the task of estimating the free energy limit.

The chordal structure is constructed as follows. In the permutation associated with the order $\bar i_1,\ldots,\bar i_m$ one associates a chord
with the two appearances of each Majorana term.  The crossing pattern of the chordal structure is then what ultimately determines 
the sign in front of the identity. This is what we compute up to leading order as follows.
It is  easy to show that for sufficiently small constant $\beta$, $m/n$ is a small constant. Thus, heuristically, one can think
of the unordered sequence $\bar i_1,\ldots,\bar i_m$ as a sub-critical ($d<1$) random (hyper) graph with only small  component sizes. The chordal structure
factorizes over connected components and thus it suffices to compute those separately for each component, each of which is  $O(1)$ in size. This is what we do using
the cavity argument. We derive a cavity type iteration for the expected value of the  sign factor for a root of a  rooted tree in terms of similar expectations
for the children of a root. The expectation is with respect to the uniform random permutation order induced by the ordered sequence $\bar i_1,\ldots,\bar i_m$.
Written in the generating function format, we obtain a functional iteration scheme, for which the answer is obtained
as a unique solution of a fixed point equation. 
The answer is then incorporated for each component leading to the expectation of the sign for the full term. One can think of these cavity iterations as tree analogues
of the Riordan-Touchard formula for chord crossing numbers, which was derived for complete graphs~\cite{riordan1975distribution} in the 1970s.

Importantly, though the component structure dominating the final term \emph{is not} the component structure of a sparse \ER because of the large deviation
effects. Loosely speaking, this is so  because the final answer is of the form $\prod b(C_j)$, where $C_j$ are the connected components and $b(C_j)\in [-1,1]$ 
is the expected sign for these components. As a result  this product is of the order $\exp(-\Theta(n))$. Consequently,
 exponentially rare random graphs can (and in fact do) contribute
dominantly to the answer if their corresponding product $\prod b(C_j)$ is exponentially larger than the one for the random (Erd{\"o}s-R\'{e}nyi) graph. 
This, however, is taken into account directly by writing the generating function of the expected sign in terms of its components. The final answer
is obtained by computing analytically  the large deviation limit of the sum product $\prod b(C_j)$, which is done using the saddle point method. 

While our computations are  proven valid only for $\beta$ small enough, the fact that our results  match the physics predictions, which have the same functional
form for all $\beta$, we conjecture that the free energy limit we compute is actually mathematically valid for all values of $\beta$. We leave this as one
of the most important open questions.

\subsubsection*{The use of large language models}
In our analysis we  have used the assistance of ChatGPT and the nature of this help is described below. 
We have derived the expression for the partition function
 (\ref{eq:sum-k-v}) directly (with no assistance of the LLM). 
 From here one can derive, as we did,  the limiting value for this expression  heuristically, by assuming that the components which
are trees with multiplicities exactly $2$ (which corresponds to the excess value $\Delta=0$) dominate this expression. 
Our next approach was to prove such domination
directly, which is technically challenging (though likely doable) by controlling the terms from components with large size and non-zero values of $\Delta$. 
Instead, ChatGPT suggested the saddle point method based on complex valued generating
function of the partition function, and worked out the necessary estimates, such as the two regimes of integration in the proof of the
 upper bound and three regimes of estimation
in the proof of the  lower bound. This is the approach we have ultimately opted for in the writeup. 
We speculate that the  preference for the saddle point method over a direct combinatorial
analysis represents a predominant approach for such computations in physics and mathematical physics literature.  Naturally, though,  this  is only our guess. 

The remainder of the paper is structured as follows. The main results with necessary prerequisites and technical background are stated in the next section.
The proof of the main result is split across  Sections~\ref{section:main-result} and~\ref{section:loose-ends}. Numerical results, including comparisons
with answers from the physics literature, are given in Section~\ref{section:Numerics}. We conclude in Section~\ref{section:open-problems} with some open problems.

\subsubsection*{Notations and conventions}
We close this section with some notational conventions. We let $[m]$ denote the set $\{1,2,\ldots,m\}$ for every positive integer $m$.
$\R,\C,\Z$ and $\Z_+$ stand for  the set of all real, complex, integer and non-negative integer values, respectively. 
$\Tr$ denotes the trace operator and $\tr=2^{-{n\over 2}}\Tr$ is the normalized trace (so that $\tr(\id)=1$ for the identity operator $\id$ in $\C^{2^{n\over 2}}$). 

Next we introduce some graph-theoretic notations.
A simple $q$-uniform hypergraph, or just simple hypergraph is a pair $(V,E)$ where $V$ is the set of nodes and $E$ is a  subset of the set of 
all ${V \choose q}$ cardinality $q$ subsets of $V$. 
For two hyperedges $e_1,e_2\in E$ we  write $e_1\sim e_2$ when $e_1\cap e_2\ne \emptyset$.
With some abuse of notation we let $\binom{V}{q}$ denote the set of all $q$-size subsets of $V$.
A multi $q$-uniform hypergraph, or just multi hypergraph, is a pair $(V,\bar r)$ where $\bar r\in \Z_+^{|V| \choose q}$. For each subset $e\in \binom{V}{q}$ 
$r_e\in \Z_+$ is called the multiplicity of the hyperedge $e$. $m(V,E)=\sum_{e\in \binom{V}{q}} r_e$ denotes the total multiplicity of the
multi hypergraph $(V,E)$.
Each multi hypergraph naturally induces a simple hypergraph where hyperedges are included
iff $r_e\ge 1$. Let $\cH^{\rm sim}([m])$ denote the set of all simple hypergraphs on the node set $[m]$. 
$\cH([m])$ denotes the  set of all multi hypergraphs on the node set $[m]$ with even  multiplicities, namely $r_e\in 2\Z_+$ for all $e$. 
Note that while $\cH^{\rm sim}([m])$ is a finite set for every $m$, the set $\cH([m])$ is infinite. 

A simple hypergraph $(V,E)$ is connected if for every $u,v\in V$ there exists a path from $u$ to $v$. Namely, there exists a sequence of hyperedges $e_1,\ldots,e_r$
such that $u\in e_1,v\in e_r$ and $e_i\cap e_{i+1}\ne\emptyset, i\in [r-1]$. A multi hypergraph is connected if its induced simple hypergraph is connected. 
Every simple and multi hypergraph $H$ admits a unique decomposition into 
connected components $H_1,\ldots,H_C$ so that $H$ is the disjoint union of $H_i$ and each $H_i$
is a connected hypergraph.
A simple hypergraph $(V,E)$ is  called a hypertree or simply a tree 
if it is a connected hypergraph with precisely $(|V|-1)/(q-1)$ hyperedges. 
It is easy to verify that every two hyperedges of a hypertree intersect in at most one node.
The set of all simple hypertrees 
on the node set $[m]$ is denoted by $\cT^{\rm sim}_m$. A singleton, namely a graph on one node with no hyperedges is defined to be connected by default, so that $\cT_{1}^{\rm sim}=\{1\}.$
A special role in our derivation 
will be played by multi hypertrees where all multiplicities are $r_e=2$, so that the total number of hyperedges is $2(m-1)/(q-1)$. We denote the set of
such hypertrees by $\cT_{m,2}$. Given $m\ge 1,\Delta\ge 0$, $\cH_{m,\Delta}\subset \cH([m])$ denotes the set of all connected multi hypergraphs  with 
precisely $2(m-1+\Delta)/(q-1)$ hyperedges. Note that $\cH_{v,0}=\cT_{v,2}$. We call $\Delta$ the excess of $H$ for $H\in\cH_{m,\Delta}$.

\section{Main result. Free energy limit}
The  SYK model is defined on the Hilbert space  $\C^{2^{n\over 2}}$ via 
the so-called Majorana  operators $\psi_i, i\in [n]$ on $\C^{2^{n\over 2}}$ satisfying 
\begin{align}\label{eq:anti-commutation}
\{\psi_i, \psi_j\}\triangleq \psi_i\psi_j+\psi_j\psi_i=\delta_{ij}. 
\end{align}
Here $\delta_{ij}$ is the identity operator $\id$ on $\C^{2^{n\over 2}}$,  when $i=j$, and $0$ otherwise, 
so that $\psi_i^2={1\over 2}\id$. These operators can be explicitly constructed using standard Pauli 
operators $X,Y,Z$ by setting $\psi_{2i+1}={1\over \sqrt{2}}(\otimes_{1\le k\le i}Z) \otimes X_{i+1}\otimes \id$
and $\psi_{2i+2}={1\over \sqrt{2}}(\otimes_{1\le k\le i}Z) \otimes Y_{i+1}\otimes \id$, for $0\le i\le n/2-1$. 
This explicit construction though will not be necessary for us. We fix a positive even integer $q$, the interaction parameter. 
For every $I=(i_1,\ldots,i_q)\in \binom{[n]}{q}, 1\le i_1<\cdots<i_q \le n$ we define $ \psi_I \triangleq i^{q/2}\psi_{i_1}\cdots \psi_{i_q}$. In particular $\psi_I^2=2^{-q}\id$,
which is an easy implication of (\ref{eq:anti-commutation}).
The Hamiltonian of the SYK model with $q$-body interaction is then defined as
the Hermitian operator
\begin{align}\label{eq:H-definition}
    H_{\rm SYK} \triangleq \sum_{I \in \binom{[n]}{q}} J_I \psi_I,
\end{align}
where $(J_I, I\in \binom{[n]}{q})$ is a sequence of i.i.d. random variables defined by $J_I=\sigma \bar J_I$ where $\sigma^2 = \frac{(q-1)!}{n^{q-1}}$
and $\bar J_I$ are Rademacher random variables, independent across $I\in \binom{[n]}{q}$. 
I.e. $\bar J_I=\pm 1$ with equal probabilities $1/2$.
By the universality property established in~\cite{anschuetz2025bounds} 
the free energy limit (defined below) will be the same for the case when $\bar J_I$ are i.i.d. standard normal
random variables, which is a more common choice in defining the SYK Hamiltonian.
For every fixed inverse temperature $\beta\in\R$ we define the associated (random) partition function 
$Z_{n,q}(\beta) \triangleq \tr(e^{\beta H_{\rm SYK}})$, where we recall that $\tr=2^{-{n\over 2}}\Tr$ stands for the normalized trace operator.
The normalization by $2^{n\over 2}$ is chosen for convenience as the trace of the identity operator $\id$ in $\C^{2^{n\over 2}}$ is $2^{n\over 2}$.
 We define the associated free energy as
\begin{align*}
    F_{n,q}(\beta) \triangleq  {1\over \beta}\log Z_{n,q}(\beta). 
\end{align*}
Our main result is the computation of the free energy limit 
\begin{align*}
F_q(\beta)\triangleq \lim_{n\to \infty} {F_{n,q}(\beta) \over n},
\end{align*}
for small enough constant inverse temperature $\beta$. In particular, we  prove
 the existence of this limit (see the discussion in the introduction), and establish that it holds w.h.p. as $n\to\infty$.
In light of the ''quenched $\approx$ annealed'' result (\ref{eq:quenched-is-annealed}) mentioned earlier and in light of the 
concentration of the quenched free energy around its expectation, also verified  in~\cite{anschuetz2025bounds}, 
we instead focus on 
\begin{align*}
    f_{n,q}(\beta) \triangleq  {1\over \beta}\log \E[Z_{n,q}(\beta)].
\end{align*}
Thus our goal is proving  the existence and computing the limit 
\begin{align*}
f_q(\beta)\triangleq \lim_{n\to \infty} {f_{n,q}(\beta) \over n},
\end{align*}
as $f_q(\beta)=F_q(\beta)$ for all $\beta$.

In order to state our result formally, we begin by introducing some necessary  technical preliminaries. 
Let $\Omega=[0,1]^2$. Elements $x=(u_1,u_2)\in \Omega$ will be referred to as ''chords''.
For every  $x=(u_1,u_2), y=(v_1,v_2)\in \Omega$, define the chord intersection function
\begin{align*}
    K(x,y) = \begin{cases}
        -1 & \text{if}~u_1 < v_1 < u_2 < v_2 \text{ or } v_1 < u_1 < v_2 < u_2,\\
        1 & \mathrm{otherwise}.
    \end{cases}
\end{align*}
Namely $K(x,y)=-1$ when the chords $x$ and $y$ cross and $=1$ otherwise. Let $\nu$ be the probability law of $(\min(U_1,U_2),\max(U_1,U_2))\in\Omega$
when $U_j, j=1,2$ are generated uniformly at random from $[0,1]$, independently. 
\begin{figure}[H]
    \centering
    \includegraphics[width=\linewidth]{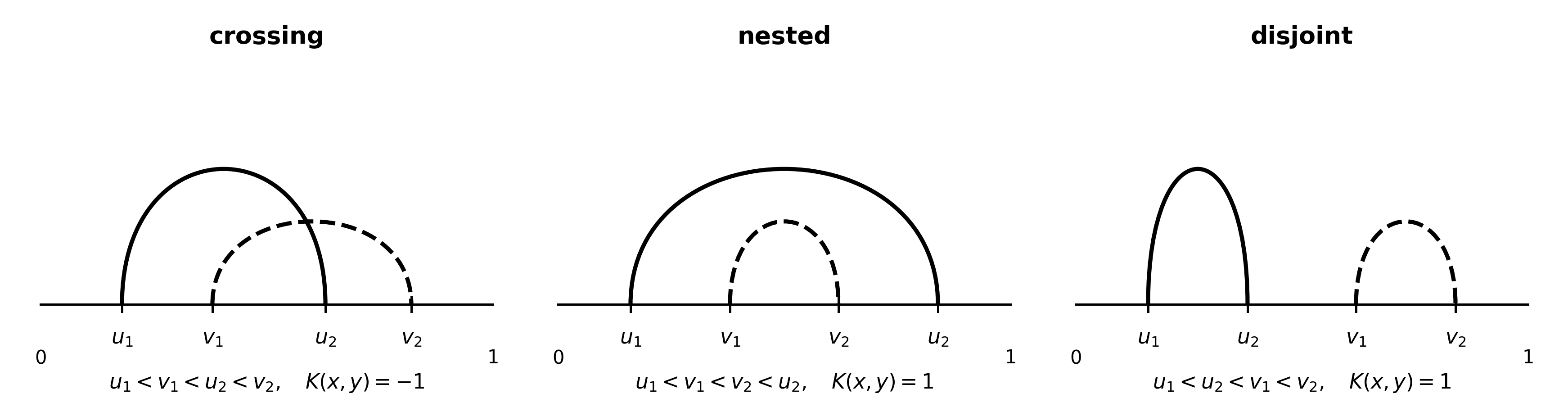}
    \caption{Chord intersection sign for two chords $x = (u_1,u_2),y=(v_1,v_2)\in \Omega.$}
    \label{fig:placeholder}
\end{figure}

On the space of continuous functions from $\{z\in \C: |z|\le 1/e^2\}\times \Omega$ to $\C$ 
define a norm $\|G\|_{e,\Omega}=\sup_{|z| \le 1/e^2, x\in\Omega} |G(z,x)|$. 
The restriction $|z|\le 1/e^2$ can be relaxed a little and is chosen for convenience.
Let $\mathcal{F}_{e,\Omega}$ 
be the subset of these functions $G$ satisfying $\|G\|_{e,\Omega}\le 1$.
On $\mathcal{F}_{e,\Omega}$ we introduce a function $H$, which 
we call the Chord-Cavity-Generator (CCG), defined implicitly via
the following fixed point equation
\begin{align}\label{eq:H-def}
 H(z,x)
        =
        z+ z\sum_{d\geq 1}\frac{1}{d!}
        \int_{\Omega^d}
        \prod_{i=1}^d
        \left[
            K(x,y_i)\,\frac{H(z,y_i)^{q-1}}{(q-1)!}
        \right]
        \prod_{1\leq i<j\leq d} K(y_i,y_j)\,
        \nu^{\otimes d}(dy_1\cdots dy_d).
\end{align}
The validity of this definition is justified by the following claim.
\begin{proposition}[Existence of CCG]\label{prop:existsCCG}
For every even $q\ge 2$ there exists $H\in\mathcal{F}_{e,\Omega}$ which is the unique solution of the equation (\ref{eq:H-def}). 
Furthermore, $H$ is analytic on the domain $\{z\in \C: |z|\le 1/e^2\}\times \Omega$.
\end{proposition}
The proof of this proposition is found in Section~\ref{section:loose-ends}.
Next we  provide a more intuitive and equivalent description of CCG via the chord crossing diagrams on trees. Consider any simple 
$q$-uniform hypergraph $H=(V,E)$ on the node set $V=V(H)$ and edge set $E=E(H)$.  
Suppose we have a chord mapping $X:E\to \Omega$. Define  
\begin{align*}
B_H(X)=\prod_{e_1,e_2\in E: |e_1\cap e_2|~{\rm odd}}K(X(e_1),X(e_2))\in \{\pm 1\}.
\end{align*}
Namely, $B_H(X)$ is $(-1)^R$, where $R$ is  the number of pairs of hyperedges which share an odd number of nodes and such that their corresponding chords cross. 
Let 
\begin{align}
b_H=\E[B_H(X)], \label{eq:b_H}
\end{align}
when $X$ is generated according to  the law $\nu^{\otimes |E|}$. Note that $| b_H|\le 1$.

Similarly, suppose we have a simple rooted hypergraph $H$ with a distinguished node $r\in V$ called a root. Fix $x\in \Omega$
and $X:E\to\Omega$. Define
\begin{align}
B_{H,r}(X,x)=\prod_{e\in E:r\in e}K(x,X(e)) 
\prod_{e_1,e_2\in E:|e_1\cap e_2|~{\rm odd}}K(X(e_1),X(e_2))\in \{\pm 1\}. \label{eq:Bhr}
\end{align}
We can think of $B_{H,r}(X,x)$ as $B_{\tilde H}(X)$ of a modified graph $\tilde H$ obtained from $H$ by adding a virtual hyperedge $\tilde e$
containing root node $r$ and $q-1$ new nodes, and such that the chord at $\tilde e$ is fixed  at value $x$. 
Let $b_{H,r}(x)=\E[B_{H,r}(X,x)]$ again when $X$ is generated according to the law $\nu^{\otimes |E|}$. We have $| b_{H,r}(x) |\le 1$ as well.
We now provide an alternative characterization of the CCG function $H(z,x)$ through the 
expected crossing functions $b$ on trees. Recall that $\cT^{\rm sim}_v$ is the set of simple hypertrees on $v$ nodes.

\begin{proposition}[CCG graph theoretically]\label{prop:CCG-graph-theoretically}
The function $H$ defined in Proposition~\ref{prop:existsCCG} satisfies
\begin{align}
 H(z,x)\triangleq 
\sum_{v=1}^\infty \frac{z^v}{v!}
    \sum_{T\in\cT_v^{\rm sim}} \sum_{r\in V(T)} b_{T,r}(x) \label{eq:H-as-tree-sum}
\end{align}
on $\{z: |z|\le 1/e^2\}\times \Omega$.
\end{proposition}
Here the  double sum  $\sum_{T\in\cT_v^{\rm sim}} \sum_{r\in V(T)}$ is set to zero  when the set $\cT_v^{\rm sim}$ 
is empty (which is the case when $q-1$ does not divide $|V|-1$).
The proof of this proposition is found in Section~\ref{section:loose-ends}.

Next we introduce two $H$-related functions, $\cR$ and $\cE$ via
\begin{align}
    \cR(z)
    &=
    z+z \sum_{d\geq 1}\frac{1}{d!}
    \int_{\Omega^d}
    \prod_{i=1}^d
    \left[
        \frac{H(z,y_i)^{q-1}}{(q-1)!}
    \right]
    \prod_{1\leq i<j\leq d} K(y_i,y_j)\,
    \nu^{\otimes d}(dy_1\cdots dy_d), \label{eq:Rdef}\\
    \cE(z)
    &=
    \frac{1}{q!}\int_{\Omega} H(z,x)^q\,\nu(dx). \label{eq:Edef}
\end{align}
Both are well-defined and analytic  on $\{z: |z|\le 1/e^2\}$ since $\sup_{z: |z|\le 1/e^2}\sup_{x\in \Omega}|H(z,x)|<\infty$ by Proposition~\ref{prop:existsCCG}. Let
 \begin{align*}
        \rho = \left(\frac{(q-1)!\beta^2}{2^{q+1}}\right)^{1/(q-1)},
    \end{align*}
and let
\begin{align*}
\mathcal{F}(z)=\rho^{-1}\left[\cR(\rho z) - (q-1)\cE(\rho z)\right], 
\end{align*}   
defined whenever both $\cR$ and $\cE$ are defined. 
As for the CCG function $H$, there is a graph-theoretic equivalent representation of $\mathcal{F}$, which 
is as follows.  For every $v\ge 1$ and $\beta$ define
\begin{align}\label{eq:mu-0}
    \mu_{v} = \frac{\sum_{T\in\cT_{v}^{\rm sim}} b_T}{v!} \rho^{v-1}.
\end{align}
We set $\mu_v=0$ when $\cT_v$ is an empty set.
We note that the set $\cT_1^{\rm sim}$ is not empty and consists of a singleton.
As it turns out $\mathcal{F}$ is the generating function for the sequence $\mu_v$. Specifically, define
\begin{align}
\Phi(z) = \sum_{v\ge 1} \mu_{v} z^v. \label{eq:Phi}
\end{align}
Our third  technical preliminary is as follows.
\begin{proposition}\label{prop:t-star}
There exists small enough $\beta^*>0$ such that for all $\beta\in \R$ satisfying $|\beta|\le \beta^*$, $\Phi$ is analytic on $\{z: |z|\le 2\}$ .
For all such $\beta$, there exists a unique solution $t^*=t^*(\beta)\in \R$ of 
\begin{align}\label{eq:t^*}
t^* \dot\Phi(t^*) = 1,
\end{align}
in the range  $(1/2,3/2)$.
\end{proposition}
The constant $2$ in $\{z: |z|\le 2\}$ is a bit arbitrary and can be made as large as desired by making $\beta^*$ sufficiently small.
The importance of the property $t^*\in (1/2,3/2)$ will be revealed only in  proof details.  
We now relate $\mathcal{F}$ to $\Phi$. 
\begin{proposition}\label{prop:t-star-2}
For $\beta^*>0$ sufficiently small and $|\beta|\le \beta^*$, $\cF(z)$ is analytic on $\{z: |z| \le 2\}$ and 
\begin{align}
\Phi(z) = \mathcal{F}(z). \label{eq:Phi-is-mathcalF}
\end{align}
\end{proposition}
Note that provided $\beta$ is small enough $|\rho z|\le 1/e^2$ when $|z|\le 2$ so that $\cF$ is indeed analytic on $|z|\le 2$.
We are now in a position to state our main result.

\begin{theorem}\label{theorem:main-result}
For every even $q$ there exists small enough $\beta^*=\beta^*(q)>0$ such that for all $\beta$ satisfying $|\beta |\le \beta^*$, 
the following convergence holds:
\begin{align}
\lim_{n\to\infty} \frac{f_{n,q}(\beta)}{n}={\cF(t^*) -1 - \log t^* \over \beta} \triangleq f_q(\beta), \label{eq:main-result}
\end{align}
where $t^*$ is defined in (\ref{eq:t^*}).
\end{theorem}
We compute $f_q$ numerically and we find that $f_q(\beta)>0$ for all $0<\beta\le \beta^*$. We leave
it as an open question to verify positivity analytically.
A simple consequence of the fact $f_q(\beta)>0$ for some $\beta>0$ is that the ground state energy is of the order $\Theta(n)$ w.h.p. 
Indeed, let $\lambda_{\max}$ be the ground state energy of the Hamiltonian $H_{\rm SYK}$, namely the largest eigenvalue of $H_{\rm SYK}$. As the 
dimension of the space is $2^{n\over 2}$ we have
$Z_{n,q}\le 2^{-{n\over 2}} 2^{n\over 2}\exp(\beta \lambda_{\max})=\exp(\beta\lambda_{\max})$ leading to
\begin{align*}
F_{n,q}(\beta)\le {1\over\beta}\log\left(\exp(\beta\lambda_{\max})\right)\le \lambda_{\max}.
\end{align*}
The convergence of $F_{n,q}$ w.h.p. to $f_q(\beta)>0$ implies $\lambda_{\max}=\Theta(n)$ w.h.p.

\section{Proof of the main result Theorem~\ref{theorem:main-result}}\label{section:main-result}

\subsection{Power series expansion of the free energy}
Given any two superindices $I,J\in\binom{[n]}{q}$, applying (\ref{eq:anti-commutation}) we have  $\psi_I\psi_J=(-1)^{|I \cap J|}\psi_J\psi_I$.
Suppose we have any ordered sequence $(I_1,\ldots,I_m)\in \binom{[n]}{q}^m$, such that each superindex
$I$ appears an even number of times in this sequence. That is, for every $k\in [m]$ there exists an odd number of  indices $[m]\ni r\ne k$ such that $I_k=I_r$. Then
from $\psi_i^2={1\over 2}\id$ we have $\psi_I^2={1\over 2^q}\id$ (recall that $q$ is even) and 
$\prod_{k\in [m]} \psi_{I_k}=\pm 2^{-qm/2}\id$,
where the sign of the expression is determined by the ordering $(I_1,\ldots,I_m)$. We denote this multiplier $\pm 1$ by $B(I_1,\ldots,I_m)$ so that
\begin{align}
\prod_{k\in [m]} \psi_{I_k}= B(I_1,\ldots,I_m) 2^{-qm/2}\id. \label{eq:psi-product}
\end{align}
Anticipating a forthcoming derivation, the function $B(I_1,\ldots,I_m)$ will be related
to $B_H(X)$ introduced earlier, when the multiplicity of each superindex $I$ is exactly two.

Consider any $q$-uniform multi-hypergraph $H=([n],\bar m)$ on the node set $V=[n]$. For each $I\in \binom{[n]}{q}$ its multiplicity is denoted by $m_I$
and $\bar m=(m_I, I\in \binom{[n]}{q})$.
Let $w(H)$ be the set of all ordered sequences $(I_1,\ldots,I_m)$ which induce precisely the multi-hypergraph $H$. That is $w(H)$
is simply the set of all ordered sequences $(I_1,\ldots,I_m)$ such that each $I$ appears exactly $m_I$ times in this sequence.
In particular, $m=\sum_I m_I$. Let 
\begin{align}
b(H)=\sum_{(I_1,\ldots,I_m)\in w(H)}  B(I_1,\ldots,I_m). \label{eq:B(H)}
\end{align}
The values  $b(H)$   will be the key building block in our proof. We note that they factorize
over connected components. That is, let $H_1,\ldots,H_R$ denote the connected components of a multi hypergraph $H$.
Recall that $m=m(H)$ denotes the  sum of edge multiplicities of a multi hypergraph $H$. Then
\begin{align}
b(H) &={m(H) \choose m(H_k) , k\in [R]}\prod_{1\le k\le R} b(H_k). \label{eq:b-factorize}
\end{align}
This holds since operators $\psi_I$ commute  over disjoint superindices $I$.
Returning to the partition function,  we have
\begin{align*}
\E[Z_{n,q}]  = \E[\tr(e^{\beta H})] = \sum_{r \geq 0} \frac{\beta^{2r}}{(2r)!} \E[\tr(H^{2r})].
\end{align*}
Here we use the fact that the expectation for odd $r$ vanishes, since odd moments of $\bar J_I$  vanish. We expand $ \E[H^{2r}]$ using  (\ref{eq:psi-product}) 
and the fact that $\E[\bar J_I^m]=1$ for even $m$ and $\E[\bar J_I^m]=0$  for odd $m$. We obtain 
\begin{align*}
\E[\tr(H^{2r})]&=\sum_{I_1,\ldots,I_{2r}\in \binom{[n]}{q}}\E[\prod_{k\in [2r]}J_{I_k}]\tr\left(\prod_{k\in [2r]}\psi_{I_k}\right) \\
&= \sum_{I_1,\ldots,I_{2r}\in \binom{[n]}{q}}\E[\prod_{k\in [2r]} J_{I_k}]B(I_1,\ldots,I_{2r}) 2^{-r q} \\
&= \sum_H 2^{-q m(H)/2}\sigma^{m(H)}b(H),
\end{align*}
where the sum is over all multi hypergraphs $H$ such that  $m(H)=2r$ and all multiplicities
$m_I(H)$ are even. 
In light of  (\ref{eq:b-factorize})   we obtain
\begin{align*}
\E[\tr(H^{2r})]&= (2r)!\sum_H\prod_{1\le k\le C} {1\over  m(H_k)  ! }2^{-q m(H_k)/2}\sigma^{m(H_k)}b(H_k),
\end{align*}
where $H_1,\ldots,H_C$ is the connected component decomposition of $H$.
Here we used $m(H)=\sum_k m(H_k)$.
Recall that  given a subset $C\subset [n]$,  $\cH(C)$ denotes  the set of all connected multi-hypergraphs $H$ with the node set $C$
with even multiplicities. Then we can rewrite the above as
\begin{align*}
\E[\tr(H^{2r})]&=\sum_{C_0,C_1,\ldots,C_R} \sum_{H_k\in\cH(C_k),k\in [R]}(2r)!\prod_{k\in [R]}{1\over  m(H_k)!}2^{-q m(H_k)/2}\sigma^{m(H_k)} b(H_k),
\end{align*}
where the first sum is over all partitions of the node set $[n]$ into subsets $C_0,C_1,\ldots,C_R$, and the second sum is over all 
tuples $H_1,\ldots,H_R$ with $H_k\in\cH(C_k)$ such that $\sum_{k\in [R]} m(H_k) =2r$.  In particular $\cup_{1\le k\le R}C_k$
is the union of the supports of the multi hypergraphs $H_k$, and $C_0=[n]\setminus \cup_{1\le k\le R} C_k$.
Combining and incorporating $\beta$ into the product terms, 
we obtain
\begin{align*}
\E[Z_{n,q}]  &= 
 \sum_{C_0,C_1,\ldots,C_R} \sum_{H_k,k\in [R]}\prod_{k\in [R]}{1\over  m(H_k)!}(\beta2^{-q/2}\sigma)^{m(H_k)}b(H_k).
 \end{align*}
Given  $v\ge 1,\Delta\ge 0$ recall that $\cH_{v,\Delta}\subset \cH([v])$ is the subset of $\cH([v])$ consisting of multi hypergraphs $H$ with precisely 
$m(H)=2(v-1+\Delta)/(q-1)$ hyperedges and $\Delta$ is the excess of $H$.
We then rewrite the above as
\begin{align*}
\E[Z_{n,q}]  &= 
 \sum_{C_0,C_1,\ldots,C_R} \sum_{\Delta_1,\ldots,\Delta_R\ge 0}
 \prod_{k\in [R]}\left(\sum_{H\in \cH_{|C_k|,\Delta_k}}{1\over  m(H)  !}(\beta2^{-q/2}\sigma)^{m(H)}b(H)\right).
 \end{align*}
Introduce  
\begin{align}\label{eq:mu}
    \mu_{v,\Delta} = {n^{v-1}\over v!}\sum_{H\in \cH_{v,\Delta}}{1\over  m(H)!}(\beta 2^{-q/2}\sigma)^{m(H)}b(H).
\end{align}
The extra factor $n^{v-1}/v!$ is introduced for later convenience. 
We stress a distinction between $\mu_{v,\Delta}$ and $\mu_v$ defined in (\ref{eq:mu-0}).
As will be established later, $\mu_v=\mu_{v,0}$. 
We also define $\mu_{v,\Delta}=0$ when $\cH_{v,\Delta}$ is an empty set. For example
this is the case 
when  $2\le v<q$ since
the graph $H$ cannot contain any hyperedges and thus is not connected.
When $v=1,\Delta=0$,
which corresponds to a (nominally) connected graph on a single node, we set
$\mu_{1,0}=1$. 
Then
\begin{align*}
\E[Z_{n,q}]  = 
\sum_{C_0,C_1,\ldots,C_R} \sum_{\Delta_1,\ldots,\Delta_R\ge 0}\prod_{k\in [R]}(|C_k|!) n^{-({|C_k|-1})}\mu_{|C_k|,\Delta_k}.
\end{align*}
Fix any set of integers $(k_{v,\Delta}, v\ge 1,\Delta\ge 0)$ satisfying
\begin{align}
\sum_{v,\Delta} vk_{v,\Delta}&=n \label{eq:sum-v}.
%{n\over q-1}+\sum_{v,\Delta} k_{v,\Delta}{\Delta-1 \over q-1}&=2r \label{eq:sum-v-Delta}
\end{align} 
Consider the set of  partitions $(C_{v,\Delta,j},v\ge 1,\Delta\ge 0, j\in [k_{v,\Delta}])$  of $[n]$ such that $|C_{v,\Delta,j}|=v$.  
The total number of such partitions is 
\begin{align}
    \frac{n!}{\prod_{v,\Delta}(v!)^{k_{v,\Delta}} k_{v,\Delta}!}.
\end{align}
These partitions represent partitions $C_0,C_1,\ldots,C_R$ where 
$C_{v,\Delta,j}, j\in [k_{v,\Delta}]$ support graphs in $\cH_{v,\Delta}$ and $C_0$
represent nodes in $[n]$ not supporting any graphs. 
Then after cancellation of $v!$ and $|C_k|!$ we obtain
\begin{align}
\E[Z_{n,q}]  &=n!\sum_{k_{v,\Delta}}\prod_{v,\Delta}{1\over n^{(v-1)k_{v,\Delta}}}{1\over k_{v,\Delta}!}\mu_{v,\Delta}^{k_{v,\Delta}} \notag \\
&=n!\sum_{k_{v,\Delta}}{1\over n^{\sum_{v,\Delta}(v-1)k_{v,\Delta}}}\prod_{v,\Delta}{1\over k_{v,\Delta}!}\mu_{v,\Delta}^{k_{v,\Delta}} \label{eq:sum-k-v}
\end{align}
where the sum is over all $(k_{v,\Delta}, v\ge 1, \Delta\ge 0)$ satisfying (\ref{eq:sum-v}). In particular, $k_{1,0}$ represents the number of nodes
not supporting graphs, hence the choice $\mu_{1,0}=1$.

Our next goal is to represent this expression as coefficients of a generating function. 
 Introduce functions $M_0,\Delta M:\C\to \C \cup \{\infty\}$ defined by
\begin{align*}
M_0(z)&=\exp\left(n \sum_{v\ge 1}\mu_{v,0}z^{v}\right), \\
\Delta M(z)&=\exp\left(n \sum_{1\le v\le n,\Delta\ge 1}\mu_{v,\Delta}z^{v}\right),
\end{align*}
We note that while in the first expression the range of $v\ge 1$ is unrestricted, in the second expression it is restricted to $v\le n$. 
This will be dictated by the fact that the first will be analytic for the infinite sum $\sum_{v\ge 1}$ for sufficiently bounded $|z|$, while for the second expression analyticity will hold  only for the final 
sum $\sum_{1\le v\le n}$. Let $M(z)=M_0(z)\Delta M(z)$.
\begin{proposition}\label{prop:generating-expansion}
There exists $\beta^*>0$ small enough such that for all $\beta$ satisfying $|\beta|\le \beta^*$,  $M_0,\Delta M$ are analytic on $\{z: |z|\le 2\}$.
Furthermore, let $M_n$ be the coefficient of $M$ in front of $z^n$, so that $M(z)=\sum_{n\ge 0}M_nz^n$. Then $\E[Z_{n,q}]={n!  \over n^n}M_n$.
\end{proposition}

\begin{proof}
The proof of analyticity is delayed and will be subsumed by stronger claims, in particular Lemma~\ref{lem:treebounds-0} covering 
the analyticity of $M_0$ and Lemma~\ref{lem:nontrees-0} covering the analyticity of $\Delta M$.
We now establish the second part. We have
\begin{align*}
M(z)&= \prod_{v\ge 1,\Delta\ge 0}\exp\left(n \mu_{v,\Delta}z^{v}\right),
\end{align*}
where for $\Delta\ge 1$ the product is restricted to $v\le n$. Then
\begin{align*}
M(z)
&=
\prod_{v\ge 1,\Delta\ge 0}\left(\sum_{k\ge 0} {\left(n\mu_{v,\Delta}z^{v}\right)^k\over k!}\right) \\
&=
\sum_{k_{v,\Delta}}\prod_{v\ge 1,\Delta\ge 0}{\left(n\mu_{v,\Delta}z^{v}\right)^{k_{v,\Delta}}\over k_{v,\Delta}!},
\end{align*}
where as before, terms with $\Delta\ge 1$ are restricted to $v\le n$.
The coefficient of $z^n$ is the sum above set to $z=1$ and restricted to tuples $(k_{v,\Delta}, v\ge 1,\Delta\ge 0)$ satisfying (\ref{eq:sum-v}).
With this restriction 
\begin{align*}
\prod_{v\ge 1,\Delta\ge 0} n^{k_{v,\Delta}}=\prod_{v\ge 1,\Delta\ge 0} n^{(v-(v-1))k_{v,\Delta}}=n^n \prod_{v\ge 1,\Delta\ge 0}{1\over n^{(v-1)k_{v,\Delta}}}.
\end{align*}
Multiplying this restricted sum by $n!\over n^n$ we obtain (\ref{eq:sum-k-v}).
\end{proof}
Our next goals are (a)  simplifying $M_0$ and analyzing its scaling limit,  and (b) establishing that the contribution of $\Delta M$ is vanishing.

\subsection{Contribution from trees}
In light of Proposition~\ref{prop:generating-expansion}, we focus on estimating $M(z)$. In this section we estimate the contribution from $M_0$
and show how it leads to the final expression (\ref{eq:main-result}) in our main result. In the following subsection we combine this estimation
with the vanishing contribution from $\Delta M$ to complete the main result.

Recall that $\cT_v^{\rm sim}$, (respectively $\cT_{v,2}$)  is the set of all  simple hypertrees  
(respectively multi hypertrees with multiplicity $2$)   on $v$ nodes. 
Recall  the definition (\ref{eq:b_H}) of $b_H$ in terms of uniformly generated chords. 
Recall also the definition (\ref{eq:B(H)}) of  $b(H)$ which is in terms of permutations of superindices. 
 We now relate $b_H$ and $b(H)$.
\begin{lemma}\label{lemma:b_v-vs-b(H)}
Given  $H\in \cH_v^{\rm sim}$ consider the corresponding multi-hypergraph in $\cH_{v,2}$ with  multiplicity $2$ for every hyperedge,  
which we denote by $H$ as well with some abuse of notation.  Then 
\begin{align}
b(H)=b_H {m(H) ! \over 2^{m(H) \over 2}}. \label{eq:bvsb}
\end{align}
\end{lemma}

\begin{proof}
$H\in \cH_v^{\rm sim}$ let $E^{\rm sim}$ be its set of (distinct) hyperedges, so that for $H$ viewed as an element of $\cH_{v,2}$
its set of hyperedges $E$ is $E^{\rm sim}$ each repeated twice.
Consider any ordered sequence $\tau=(I_1,\ldots,I_{2|E|})$ of $2|E|$ hyperedges of the multi hypergraph version of  $H$.
The total number of such sequences is $|E|!/2^{|E|/2}$.
Define $X_\tau:E^{\rm sim}\to [|E|]^2$ as follows. For each $e\in E^{\rm sim}$, 
$X(e)=(u^1,u^2)$ if the first copy of the edge $e$ in the permutation $\tau$ appears 
in position $u^1$ and the second copy in position $u^2$. Given $X_\tau$ generate a random $X:E\to \Omega$ where each $u^j$ is mapped
to a point in $[{u^j-1\over |E|},{u^j\over |E|}]$ chosen uniformly at random. Suppose $\tau$ is also chosen u.a.r. 
The distribution of $X$ is then $\nu^{\otimes |E|}$. The relation (\ref{eq:bvsb}) then follows since 
\begin{align*}
b(H)\left({m(H) ! \over 2^{m(H) \over 2}}\right)^{-1}
\end{align*}
is the expectation of $B(I_1,\ldots,I_{2|E|})$ with respect to uniformly random $\tau$. 
\end{proof}

Next we validate $\mu_{v,0}=\mu_v$ for the case of multi  hypertrees $T\in \cT_{v,2}$ where $\mu_v$ was defined in (\ref{eq:mu-0}). 
By definition we have $\mu_{1,0}=\mu_1$. For the remaining cases,
applying Lemma~\ref{lemma:b_v-vs-b(H)} we have
\begin{align*}
    \mu_{v,0} &= {n^{v-1}\over v!}\sum_{T\in \cT_{v,2}}{1\over  
    m(T)!}(\beta^2 2^{-q}\sigma^2)^{m(T)/2}b(H) \\
    &=
{n^{v-1}\over v!}\sum_{T\in \cT_{v}^{\rm sim}}(\beta^2 2^{-q-1}\sigma^2)^{m(T)}b_{T}.  
\end{align*}
For every $T\in \cT_{v}^{\rm sim}$ we have $m(T)=(v-1)/(q-1)$. Then
\begin{align*}
n^{v-1}(\sigma^2)^{m(T)}=n^{v-1}\left( { (q-1)!\over n^{q-1}}\right)^{v-1 \over q-1}=((q-1)!)^{v-1 \over q-1}.
\end{align*}
We obtain:
\begin{align*}
\mu_{v,0}&=
{\left(\beta^2 2^{-q-1}(q-1)!\right)^{v-1\over q-1}\over v!}\sum_{H\in \cT_{v}^{\rm sim}}b_{H} \\
&=
{\rho^{v-1}\over v!}\sum_{H\in \cT_{v}^{\rm sim}}b_{H} \\
&=\mu_v.
\end{align*}

Our next goal is proving Proposition~\ref{prop:t-star}.

\begin{proof}[Proof of Proposition~\ref{prop:t-star}]
We begin by establishing some bounds for $\Phi$.
\begin{lemma}\label{lem:treebounds-0}
For every $c>0$  there exists small enough $\beta^*>0$ such that for all $\beta$ with $|\beta|\le \beta^*$ 
the function $\Phi(z)$ (defined by (\ref{eq:Phi}))  is analytic on $|z| \leq 2$. Furthermore
    \begin{align}
    	\sup_{|z| \leq 2} \abs{\Phi(z) - z} \leq c, \label{eq:Phi-minus-z}\\
	 \sup_{|z| \leq 2} \abs{\dot \Phi(z) - 1} \leq c, \label{eq:Phi-dot-minus-1}\\
	 \quad \sup_{|z| \leq 2} |z|^2 \abs{\ddot\Phi(z)} \leq c. \label{eq:Phi-2nd-derivative}
    \end{align}
\end{lemma}
\begin{proof}
Since $|b_H|\le 1$ and $\left|\mathcal{T}_v^{\text {sim }}\right|=\frac{(v-1)!v^{s-1}}{s!((q-1)!)^s}$, where $s=(v-1)/(q-1)$ is the number of hyperedges.
By Stirling's approximation we have 
\begin{align}
|\mu_v| &\le (1+o(1))e^v\lr{\frac{(q-1)! \beta^2}{2^{q+1}}}^{(v-1)/(q-1)} \notag \\
&\le 
(1+o(1)) e\left(e\lr{\frac{(q-1)! \beta^2}{2^{q+1}}}^{1/(q-1)}\right)^{v-1} \notag \\
&\triangleq 
(1+o(1))eg^{v-1}(\beta).
\label{eq:abs-mu}
\end{align}
We choose $\beta^*$ small enough so that when $|\beta |\le \beta^*$, $\sum_v g^v(\beta)2^v<\infty$. Then, recalling $\mu_1=1$
and $\mu_v=0$ when $v\ge 2$ and $\cT_v$ is empty, we obtain
$|\Phi(z)|=|z+\sum_{v\ge 2} \mu_v z^v|\le 2+\sum_v g^v(\beta)2^v<\infty$, and thus $\Phi$ is analytic on $|z|\le 2$.

We now establish the claimed bounds. We have $\Phi(z)-z=\sum_{v\ge 2}\mu_v z^v$. 
Given $c>0$,  we can find $\beta^*$ small enough so that $|g(\beta^*)|$ in (\ref{eq:abs-mu}) satisfies $\sum_{v\ge 2}  (1+o(1))eg^{v-1}(\beta^*)2^v\le c$ for every $v\ge 2$.
Namely (\ref{eq:Phi-minus-z}) holds.
 
Next, $\dot \Phi(z)-1=\sum_{v\ge 2} v\mu_v z^{v-1}$. Using a very similar argument we find $\beta^*$ small enough so that (\ref{eq:Phi-dot-minus-1}) holds.
Finally, $\ddot \Phi(z)=\sum_{v\ge 2} v(v-1)\mu_v z^{v-2}$. A similar argument is used to find $\beta^*$ small enough so that (\ref{eq:Phi-2nd-derivative}) holds as well. 
\end{proof}

We return to the proof of the proposition
and show that for $\beta^*$ sufficiently small $F(t) \triangleq t \dot\Phi(t) - 1$ has a unique zero $t^*$ in the interval $(1/2, 3/2)$. 
Choose $c<1/3$. We find $\beta^*$ small enough so that by Lemma~\ref{lem:treebounds-0}
${1\over 2}(\dot\Phi\left({1\over 2}\right)-1)\le c/2$, implying ${1\over 2}\dot\Phi\left({1\over 2}\right)-1\le c/2-1/2<0$.
Similarly, ${3\over 2}(\dot\Phi\left({3\over 2}\right)-1)\ge -3c/2$, implying ${3\over 2}\dot\Phi\left({3\over 2}\right)-1\ge 1/2-3c/2>0$.
Thus $F$ has zero in the interval $[1/2,3/2]$. To show uniqueness, again using Lemma~\ref{lem:treebounds-0}, we note that for $t\in [1/2,3/2]$
\begin{align*}
\dot F(t)=\dot \Phi(t)+t\ddot \Phi(t) \ge 1-c-{1\over (1/2)}c=1-3c>0.
\end{align*}
Therefore $F$ is strictly increasing and the root is unique.
\end{proof}

\subsection{Contribution from non-trees}
We turn to analyzing $\Delta M$. Define 
\begin{align*}
B(z)=\log \Delta M(z)=n \sum_{1\le v\le n,\Delta\ge 1}\mu_{v,\Delta}z^{v}.
\end{align*}

\begin{lemma}\label{lem:nontrees-0}
There exists small enough $\beta^*>0$  so that when $|\beta|\le \beta^*$ 
	\begin{align}
		\sup_{n \geq 1} \sup_{|z| \leq 2} |B(z)| \leq 1, \quad \sup_{n \geq 1} \sup_{|z| \leq 2} |\dot B(z)| \leq 1. \label{eq:B-dot-B}
	\end{align}
\end{lemma}

\begin{proof}
We begin by establishing the following bound. For all $v\ge 1$ and $\Delta\ge 1$ we claim 
\begin{align}\label{eq:count-graphs}
\sum_{H\in\mathcal{H}_{v,\Delta}}  {| b(H)| \over m(H)!}\le (2e)^{v+\Delta} v^{v+\Delta}.
\end{align}
Recall that $\cH_{v,\Delta}=\emptyset$ when $\Delta\ge 1$ and $v<q$, and thus we may assume $v\ge q$.
To show the claimed bound it suffices to show 
\begin{align*} |\mathcal{H}_{v,\Delta}|\le  (qe)^{v+\Delta} v^{v+\Delta},
\end{align*}
since $|b(H)|\le m(H)!$. Recall that the total number of hyperedges 
in a graph $H\in\mathcal{H}_{v,\Delta}$ is $E\triangleq 2{v-1+\Delta\over q-1}$ by definition. The number of distinct possible hyperedges is ${v\choose q}$. 
Let $2m_1,\ldots,2m_{v\choose q}$ be their multiplicities (including the possibility $m_j=0$), which have to be even by definition of $\mathcal{H}_{v,\Delta}$. 
Then $2m_1+2m_2+\cdots+2m_{v\choose q}=E$, implying   
\begin{align*}
| \mathcal{H}_{v,\Delta}|={{v\choose q}+{E\over 2}-1\choose {v\choose q}-1}.
\end{align*}
Let $\bar v={v\choose q}$. We have
\begin{align*}
{\bar v +{E\over 2}-1 \choose \bar v-1}
&=
{\bar v +{E\over 2}-1 \choose {E\over 2}} \\
 &\le \left( {e (\bar v+{E\over 2})\over {E\over 2}} \right)^{{E\over 2}}  \\
&=e^{{E\over 2}}\left(1+{\bar v \over {E\over 2}}\right)^{{E\over 2}}.
\end{align*}
We have
\begin{align*}
{\bar v \over {E\over 2}}={{v\choose q} \over {v-1+\Delta \over q-1}}={v(v-1)\cdots(v-q+1) \over q!}{q-1\over v-1+\Delta}
\end{align*}
Since $\Delta\ge 1$ we have $v-1+\Delta\ge v$ and thus 
we obtain a bound
\begin{align*}
{v(v-1)\cdots(v-q+1) \over q!}{q\over v}
\le v^{q-1}.
\end{align*}
Then
\begin{align*}
e^{{E\over 2}}\left(1+{\bar v \over {E\over 2}}\right)^{{E\over 2}} &\le e^{E\over 2} \left(1+v^{q-1}\right)^{E\over 2} \\
&\le (2e)^{E\over 2} \left(v^{q-1}\right)^{{v-1+\Delta\over q-1}} \\
&= (2e)^{E\over 2} v^{(v-1+\Delta)} \\
&\le (2e)^{v+\Delta} v^{v+\Delta}.
\end{align*}
This completes the proof of (\ref{eq:count-graphs}).

We now return to the proof of the lemma. Recall that $\mu_{v,\Delta}=0$ when $\Delta\ge 1$ and $v<q$. We have
\begin{align*} 
\mu_{v,\Delta} =  
\sum_{H\in \cH_{v,\Delta}}{1\over  \left( 2{v-1+\Delta\over q-1} \right) !}{n^{v-1}\over v!}(\beta^2 2^{-q}\sigma^2)^{v-1+\Delta\over q-1}b(H).
\end{align*}
Then
\begin{align*} 
|\mu_{v,\Delta}| \le   {\sum_{H\in \cH_{v,\Delta}} |b(H)| \over   \left( 2{v-1+\Delta\over q-1} \right) !}
{n^{v-1}\over v!}(\beta^2 2^{-q}\sigma^2)^{v-1+\Delta\over q-1}.
\end{align*}
By (\ref{eq:count-graphs}) it is at most 
\begin{align*}
|\mu_{v,\Delta}|& \le (2e)^{v+\Delta} v^{v+\Delta}{n^{v-1}\over v!}(\beta^2 2^{-q}\sigma^2)^{v-1+\Delta\over q-1} \\
&=
(2e)^{v+\Delta} v^{v+\Delta}{n^{v-1}\over v!}(\beta^2 2^{-q}(q-1)!n^{-(q-1)})^{v-1+\Delta\over q-1} \\
&=
(2e)^{v+\Delta} v^{v+\Delta}{n^{v-1}\over v!}(\beta^2 2^{-q}(q-1)!)^{v-1+\Delta\over q-1}{1\over n^{v-1+\Delta}} \\
&=
(2e)^{v+\Delta} v^{v+\Delta}{1\over v!}(\beta^2 2^{-q}(q-1)!)^{v-1+\Delta\over q-1}{1\over n^{\Delta}} 
\end{align*}
Then
\begin{align*}
\max_{|z|\le 2}|B(z)|\le 
n\sum_{q\le v\le n,\Delta\ge 1}2^v |\mu_{v,\Delta}|  
&\le 
n(4e (\beta^2 2^{-q}(q-1)!)^{1\over q-1})^{v+\Delta}  v^{v+\Delta}{1\over v!}{1\over n^{\Delta}}
\end{align*}
We find universal constant $c\ge 1$ so that ${v^v\over v!}\le c^v$ for all $v\ge 1$. Denote $(4e)c (\beta^2 2^{-q}(q-1)!)^{1\over q-1}$ by $g(\beta)$.
We obtain a bound
\begin{align*}
\max_{|z|\le 2}|B(z)|\le  \left(g(\beta)\right)^{v+\Delta}v^{\Delta}{1\over n^{\Delta-1}}
\end{align*}
For every fixed $v\le n$ consider the sum 
\begin{align*}
\sum_{\Delta \ge 1} \left({g(\beta) v \over n}\right)^{\Delta-1}={1\over 1-{g(\beta) v\over n}} \le 2,
\end{align*}
provided $\beta$ small enough so that $g(\beta)v/n\le g(\beta)\le 1/2$.  
We obtain a bound
\begin{align*}
\max_{|z|\le 2}|B(z)| \le 2\sum_{q\le v\le n} \left(g(\beta)\right)^{v+1}v 
\end{align*}
which we arrange to be less than $1$ by taking $\beta$ sufficiently small.

The derivation for $\dot B(z)=n \sum_{1\le v\le n,\Delta\ge 1}\mu_{v,\Delta}vz^{v-1}$ is similar.
\end{proof}

\subsection{Proof completion}

Recalling  Proposition~\ref{prop:generating-expansion},
we now focus on asymptotics of $M_n$. We fix $\beta^*$ small enough and find the corresponding $t^*$ as in Proposition~\ref{prop:t-star}. 
By Cauchy's integral formula on $z=t^* e^{i\theta}$, we can write $M_n$ as
	\begin{align*}
		M_n &= \frac{1}{2\pi i}\oint_{|z|=t^*}\frac{M(z)}{z^{n+1}}\,dz \\
		&=\frac{1}{2\pi i}\oint_{|z|=t^*}\frac{\exp(n\Phi(z))\Delta M(z)}{z^{n+1}}\,dz \\
		&= \frac{(t^*)^{-n}}{2\pi}\int_{-\pi}^\pi \exp\left(n \Psi(\theta) + B(t^*e^{i\theta})\right)d\theta,
	\end{align*}
where we define $\Psi(\theta) = \Phi(t^*e^{i\theta}) - i\theta$ and recall that $B(z)=\log \Delta M(z)$. 
The remainder of the section is devoted to establishing the following proposition, from which the proof of Theorem~\ref{theorem:main-result} 
will be an easy consequence.
\begin{proposition}\label{prop:M_n}
The following holds
\begin{align*}
\lim_n {1\over n}\log M_{n}=-\log t^*+\Psi(0)=-\log t^*+\Phi(t^*).
\end{align*}
\end{proposition}

\begin{proof}
The proof is based on  standard saddle point method. In particular, we will show that $\theta=0$ is the unique maximizer of $\Psi$, and show that the contribution
from $B$ is negligible, due to Lemma~\ref{lem:nontrees-0}.

\subsubsection{Proof of upper bound.}
We first establish
\begin{align}
\limsup_n {1\over n}\log M_{n}\le -\log t^*+\Psi(0). \label{eq:limsup-M-n}
\end{align}
Since $t^* \dot\Phi(t^*) = 1$, we have $\dot\Psi(0) = it^*\dot\Phi(t^*) - i = 0$. Next
\begin{align}
\ddot \Psi(\theta)&=-t^*e^{i\theta}\dot \Phi(t^*e^{i\theta})-(t^*e^{i\theta})^2\ddot\Phi(t^*e^{i\theta}) \notag\\
&=-t^*e^{i\theta}+t^*e^{i\theta}\left(1-\dot \Phi(t^*e^{i\theta})\right)-(t^*e^{i\theta})^2\ddot\Phi(t^*e^{i\theta}). \label{ddot-Psi}
\end{align}
Using  bounds from Lemma~\ref{lem:treebounds-0}, which apply since $|t^*e^{i\theta}|=t^*\in [1/2,3/2]$, we obtain
\begin{align*}
 |1-\dot \Phi(t^*e^{i\theta})| &\le c, \\
|t^*e^{i\theta}|^2 |\ddot\Phi(t^*e^{i\theta})| &\le c.
\end{align*}
Thus
\begin{align*}
\Re\ddot\Psi(\theta)\le -\Re(t^*e^{i\theta})+(3/2)c+c =-t^*\cos(\theta)+(5/2)c.
\end{align*}
Assume $|\theta|\le 1$. Then $t^*\cos(\theta)\ge (1/2)\cos(1)$. We take  $c$ small enough so that $-(1/2)\cos(1)+(5/2)c\le -1/4$ and obtain
that $\Re \ddot\Psi(\theta)\le -1/4$ in the range $|\theta|\le 1$. Then in this range by Taylor expansion we obtain
\begin{align}\label{eq:theta-less-1}
\Re\Psi(\theta)\le \Psi(0)-\theta^2/8=\Phi(t^*)-\theta^2/8.
\end{align}
Next we bound $\Re\Psi(\theta)$ in the remaining range $|\theta|\ge 1$. We have
\begin{align*}
\Re \Psi(\theta) &= \Re \Phi(t^* e^{i\theta}) \\
&= \Phi(t^*) - t^*(1-\cos \theta) - \left(\Phi(t^*) - t^*\right) + \Re \left(\Phi(t^*e^{i\theta}) - t^* e^{i\theta}\right) \\
\end{align*}
Applying Lemma~\ref{lem:treebounds-0}
\begin{align*}
|\Phi(t^*) - t^*|, | \Phi(t^*e^{i\theta}) - t^* e^{i\theta} |\le c.
\end{align*}
Using $t^*\ge 1/2$, we obtain an upper bound
\begin{align*}
\Re \Psi(\theta) & \le \Phi(t^*)-(1/2)(1-\cos(1))+2c.
\end{align*}
We assume $c>0$ is small enough so that the right-hand side is at most $\Phi(t^*)-1/5$. Thus when $|\theta|\ge 1$ we obtain a bound
\begin{align}\label{eq:theta-greater-1}
\Re\Psi(\theta)\le \Psi(0)-1/5.
\end{align}
We now complete the proof of (\ref{eq:limsup-M-n}). We assume $\beta^*$ small enough so that bounds (\ref{eq:B-dot-B}) from Lemma~\ref{lem:nontrees-0}
apply. Then
\begin{align*}
M_n&=
\frac{(t^*)^{-n}}{2\pi}\int_{-\pi}^\pi \Re\left(\exp\left(n \Psi(\theta) + B(t^*e^{i\theta})\right)\right)d\theta\\
&\le
\frac{(t^*)^{-n}}{2\pi}\int_{-\pi}^\pi \exp\left( \Re\left(n\Psi(\theta)) + B(t^*e^{i\theta})\right)\right)d\theta\\
&\leq \frac{(t^*)^{-n}}{2\pi} \lr{\int_{|\theta| \leq 1} e^{n\Phi(t^*) - n\theta^2/8 + 1}d\theta 
+ \int_{|\theta| \geq 1, \theta\in [-\pi,\pi]} e^{n\Phi(t^*) - n/5 + 1} d\theta} \\
&\leq (t^*)^{-n} e^{n\Phi(t^*)}  \frac{e}{2\pi}\lr{\sqrt{\frac{8\pi}{n}} + 2\pi e^{-n/5}},
\end{align*}
and (\ref{eq:limsup-M-n}) follows.

\subsubsection{Proof of lower bound.}
We now establish
\begin{align}
\liminf_n {1\over n}\log M_{n}\ge -\log t^*+\Psi(0). \label{eq:liminf-M-n}
\end{align}
For this we  partition the integrated region $[-\pi,\pi]$ into three regions.
Introduce constant $\eta\in (0,1], L\ge 1$ to be specified later. The choice of $\eta$ will depend on $\beta^*$. The choice of $L$ will be universal.
We assume that $n$ is sufficiently large so that $1/\sqrt n \leq \eta$. 
Our regions will be $|\theta|\le L/\sqrt{n}, L/\sqrt{n}\le \theta\le 1$ and $|\theta|\ge 1$. The choice of $L$ in particular will be only relevant
to the regime $L/\sqrt{n}\le \theta\le 1$.

\begin{enumerate}
\item 
$ |\theta |\le L/\sqrt{n}$. We have from (\ref{ddot-Psi})
\begin{align*}
\ddot\Psi(0)=-1-(t^*)^2\ddot\Phi(t^*) \in \R.
\end{align*}
Applying Lemma~\ref{lem:treebounds-0} $1+\ddot\Psi(0)\in [-c,c]$. We find sufficiently small $\eta>0$ such that when $|\theta |\le \eta$ we have
$1+\Re(\ddot\Psi(\eta))\in [-2c,2c]$. Recall that $\dot\Psi(0)=0$. Then by Taylor expansion,
\begin{align*}
\Phi(t^*)-{1\over 2}(1+2c)\theta^2 \le \Re\Psi(\theta)\le \Phi(t^*)-{1\over 2}(1-2c)\theta^2.
\end{align*}
We assume $c$ is small enough so that 
\begin{align}
\Phi(t^*)-\theta^2 \le \Re\Psi(\theta)\le \Phi(t^*)-{1\over 8}\theta^2. \label{eq:Phi-double-tail}
\end{align}
Next we control the imaginary part of $\Psi$. Since $\dot \Psi(0)=0$ and $\ddot\Psi(0)$ does not have an imaginary part we have
\begin{align*}
\abs{\Im(\Psi(\theta) - \Phi(t^*))} \leq |\theta|^3 \cdot \frac{1}{6} \sup_{|w| \leq \eta} \abs{{d^3\Psi(w)\over dw^3}}\triangleq C|\theta|^3.
\end{align*}
Here $C$ depends only on the choice of $\beta^*$ and $\eta$. 
Applying Lemma~\ref{lem:nontrees-0}
\begin{align*}
\abs{\Im(B(t^*e^{i\theta}) - B(t^*))} \leq \abs{t^* e^{i\theta} - t^*}\sup_\theta |\dot B(t^*e^{i\theta})| \leq \frac{3}{2}2|\theta |  = 3|\theta|.
\end{align*}

Assume $n$ sufficiently large such that $L/\sqrt{n}\le \eta$ and
\begin{align}
\frac{CL^3 + 3L}{\sqrt n} \leq \frac{\pi}{3},
\end{align}
Then for $|\theta|\le L/\sqrt{n}$
\begin{align*}
\abs{\Im\lr{n\left(\Psi(\theta)  - \Phi(t^*)\right)+B(t^*e^{i\theta}) - B(t^*)}} \leq Cn|\theta|^3+3|\theta |
\le
\frac{\pi}{3}.
\end{align*}
Using $\cos(\pi/3)=1/2$ we have
\begin{align*}
\Re\int_{|\theta| \leq L/\sqrt n} e^{n\Psi(\theta) + B(t^* e^{i\theta})} d\theta &\geq 
\frac{1}{2} \int_{|\theta| \leq L/\sqrt n} e^{\Re n\Psi(\theta) + \Re B(t^* e^{i\theta})} d\theta.
\end{align*}
By Lemma~\ref{lem:nontrees-0} $|B(t^*)-B(t^*e^{i\theta})|\le |\theta|<1$. Then 
$|\Re(B(t^*e^{i\theta}))-\Re(B(t^*))|<1$, and since $B(t^*)$ is real valued
\begin{align*}
\Re(B(t^*e^{i\theta}))\ge B(t^*)-1\ge -2,
\end{align*}
where in the last inequality Lemma~\ref{lem:nontrees-0} is invoked again. Then
\begin{align*}
\int_{|\theta| \leq L/\sqrt n} e^{\Re n\Psi(\theta) + \Re B(t^* e^{i\theta})} d\theta \ge
e^{-2}\int_{|\theta| \leq L/\sqrt n} e^{\Re n\Psi(\theta) } d\theta.
\end{align*}
Applying the left part of (\ref{eq:Phi-double-tail}) the right-hand side above is at least
\begin{align}
e^{-2}e^{n\Phi(t^*)}\int_{|\theta| \leq L/\sqrt n} e^{-n\theta^2} d\theta=e^{n\Phi(t^*)}{e^{-2}\over \sqrt{n}}\int_{-L}^Le^{-u^2} du
\triangleq K_n. \label{eq:inner-integral}
\end{align}
We note 
\begin{align}
\lim_n {1\over n}\log K_n=\Phi(t^*). \label{eq:lim-Kn}
\end{align}

\item $L/\sqrt{n}\le \theta\le 1$.
For this regime, we have by (\ref{eq:theta-less-1}) and  Lemma~\ref{lem:nontrees-0} 
	\begin{align*}
		\int_{L/\sqrt n \leq |\theta| \leq 1} \abs{e^{n\Psi(\theta) + B(t^* e^{i\theta})}} d\theta 
		&\leq e^{1 + n \Phi(t^*)} \int_{L/\sqrt n \leq |\theta| \leq 1} e^{-n\theta^2/8} d\theta \\
		&\leq \frac{e}{\sqrt n} e^{n\Phi(t^*)} \int_{|u| \geq L} e^{-u^2/8} du.
	\end{align*}
We choose $L$ large enough so that this upper bound is at most $(1/3)K_n$. Namely
\begin{align*}
 \frac{e}{\sqrt n} e^{n\Phi(t^*)} \int_{|u| \geq L} e^{-u^2/8} du
\le
{1\over 3}e^{n\Phi(t^*)}{e^{-2}\over \sqrt{n}}\int_{-L}^Le^{-u^2} du.
\end{align*}
The choice of $L$ here is universal.

\item $\theta\ge 1$. For the outer integral, we have by Lemma~\ref{lem:nontrees-0} and (\ref{eq:theta-greater-1})  that
\begin{align*}
\int_{|\theta| \geq 1}\Re(\exp(n\Psi(\theta)+B(t^*e^{i\theta})))d\theta 
\le
2\pi e e^{n\Phi(t^*) - n/5}.
\end{align*}
We choose $n$ large  enough so that again this is at most $(1/3)K_n$.
Combining all these bounds, the sum of three integrals is at least $(1/3)K_n$. Applying (\ref{eq:lim-Kn}) we complete the proof of (\ref{eq:liminf-M-n})
and thus the proof of Proposition~\ref{prop:M_n},
\end{enumerate}
\end{proof}

\begin{proof}[Proof of Theorem~\ref{theorem:main-result}]
We have completed the proof of the main result when $\Phi$ is defined as in (\ref{eq:Phi}). 
In order to complete the proof we need to validate Propositions~\ref{prop:existsCCG} and~\ref{prop:t-star-2}.  We will also prove
Proposition~\ref{prop:CCG-graph-theoretically}, which will be a step in proving Proposition~\ref{prop:t-star-2}.
These propositions are established in the next section.
\end{proof}

\section{Loose ends. Proof of Propositions~\ref{prop:existsCCG},\ref{prop:CCG-graph-theoretically},\ref{prop:t-star-2}}
\label{section:loose-ends}

\begin{proof}[Proof of Proposition~\ref{prop:existsCCG}]
Consider an operator $\cT$ acting on functions $G\in \mathcal{F}_{e,\Omega}$ defined by
\begin{align*}
 \cT(G)(z,x)=z+z\sum_{d\geq 1}\frac{1}{d!}\int_{\Omega^d}\prod_{i=1}^d
 \left[K(x,y_i)\,\frac{G^{q-1}(z,y_i)}{(q-1)!}\right]
\prod_{1\leq i<j\leq d} K(y_i,y_j)\,\nu^{\otimes d}(dy_1\cdots dy_d).
\end{align*}
While the range of $\cT(G)$ could potentially include $\infty$ if the infinite sum is not absolutely convergent, we show below that this is not the case.
\begin{lemma}\label{lemma:T-bounded-contraction}
For every $G\in \mathcal{F}_{e,\Omega}$, we have $\mathcal{T}(G)\in \mathcal{F}_{e,\Omega}$. Furthermore, for every $G_1,G_2\in\mathcal{F}_{e,\Omega}$,
\begin{equation}
\| \mathcal{T}(G_1)-\mathcal{T}(G_2)\|_{e,\Omega}
\le e^{-1}\| G_1-G_2\|_{e,\Omega}. \label{eq:T-contraction}
\end{equation}
Namely, $\mathcal{T}$ is a contraction on the space $\mathcal{F}_{e,\Omega}$ with respect to the metric induced by the norm $\|\cdot\|_{e,\Omega}$.
\end{lemma}

\begin{proof}
Let
\begin{equation*}
D_e=\{z\in\mathbb C: |z|\le e^{-2}\}.
\end{equation*}
Fix $G\in \mathcal{F}_{e,\Omega}$. We first show that $\mathcal T(G)$ is well-defined and belongs to $\mathcal F_{e,\Omega}$.

For $z\in D_e$ and $x\in\Omega$, using $|K|=1$, $\nu(\Omega)=1$, and $\|G\|_{e,\Omega}\le 1$, we have
\begin{align*}
|\mathcal{T}(G)(z,x)|
&\le |z|
+
|z|\sum_{d\geq 1}\frac{1}{d!}
\int_{\Omega^d}
\prod_{i=1}^d
\frac{|G(z,y_i)|^{q-1}}{(q-1)!}
\,\nu^{\otimes d}(dy_1\cdots dy_d)\\
&\le |z|
+
|z|\sum_{d\geq 1}\frac{1}{d!}\\
&= |z|e \le e^{-1} \le 1.
\end{align*}
The same estimate shows that the defining series for $\mathcal T(G)$ converges absolutely and uniformly on $D_e\times\Omega$.

It remains to check continuity. For each fixed $d$, the corresponding integral is continuous in $(z,x)$ by dominated convergence. Indeed, $G$ is continuous, the integrand is uniformly bounded by an integrable function, and for every fixed $x\in\Omega$ the discontinuity set of the factors involving $K(x,y_i)$ has $\nu^{\otimes d}$-measure zero. Since the series converges uniformly on $D_e\times\Omega$, $\mathcal T(G)$ is continuous on $D_e\times\Omega$. Therefore $\mathcal T(G)\in\mathcal F_{e,\Omega}$.

Now fix $G_1,G_2\in\mathcal F_{e,\Omega}$ and set
\begin{equation*}
\Delta=\|G_1-G_2\|_{e,\Omega}.
\end{equation*}
For fixed $z\in D_e$, $x\in\Omega$, $d\ge1$, and $y_1,\ldots,y_d\in\Omega$, define
\begin{equation*}
A_i
=
K(x,y_i)\frac{G_1(z,y_i)^{q-1}}{(q-1)!},
\qquad
B_i
=
K(x,y_i)\frac{G_2(z,y_i)^{q-1}}{(q-1)!}.
\end{equation*}
Since $G_1,G_2\in\mathcal F_{e,\Omega}$, we have $|A_i|\le1$ and $|B_i|\le1$. By the telescoping identity,
\begin{align*}
\prod_{i=1}^d A_i-\prod_{i=1}^d B_i
&=
\sum_{k=1}^d
\left(\prod_{i<k}A_i\right)
(A_k-B_k)
\left(\prod_{i>k}B_i\right).
\end{align*}
Therefore,
\begin{align*}
\left|
\prod_{i=1}^d A_i-\prod_{i=1}^d B_i
\right|
&\le
\sum_{k=1}^d |A_k-B_k|.
\end{align*}
For each $k$,
\begin{align*}
|A_k-B_k|
&=
\frac{\left|G_1(z,y_k)^{q-1}-G_2(z,y_k)^{q-1}\right|}{(q-1)!}\\
&\le
\frac{q-1}{(q-1)!}
|G_1(z,y_k)-G_2(z,y_k)|\\
&\le
\Delta.
\end{align*}
Here we used the elementary inequality
\begin{equation*}
|a^m-b^m|
\le
m|a-b|,
\qquad |a|,|b|\le1,
\end{equation*}
with $m=q-1$, together with $(q-1)/(q-1)!\le1$ for $q\ge2$.

The linear term $z$ cancels when subtracting $\mathcal T(G_1)$ and $\mathcal T(G_2)$. Hence, using again $|K|=1$,
\begin{align*}
|\mathcal T(G_1)(z,x)-\mathcal T(G_2)(z,x)|
&\le
|z|
\sum_{d\ge1}\frac{1}{d!}
\int_{\Omega^d}
\left|
\prod_{i=1}^d A_i-\prod_{i=1}^d B_i
\right|
\,\nu^{\otimes d}(dy_1\cdots dy_d)\\
&\le
|z|
\sum_{d\ge1}\frac{d}{d!}\Delta\\
&=
|z|e\,\Delta\\
&\le
e^{-1}\Delta.
\end{align*}
Taking the supremum over $z\in D_e$ and $x\in\Omega$ gives
\begin{align*}
\|\mathcal T(G_1)-\mathcal T(G_2)\|_{e,\Omega}
&\le
e^{-1}\|G_1-G_2\|_{e,\Omega}.
\end{align*}
This proves \eqref{eq:T-contraction}.
\end{proof}

Returning to the proof of the proposition, the existence and uniqueness of the solution to the fixed point equation \eqref{eq:H-def} now follow from the Banach fixed point theorem, applied to the complete metric space $\mathcal F_{e,\Omega}$ with the metric induced by $\|\cdot\|_{e,\Omega}$. Equivalently, the unique fixed point is the uniform limit of the sequence
\begin{equation*}
H_0\equiv 1,\qquad H_{n+1}=\mathcal T(H_n),\quad n\ge0.
\end{equation*}
\end{proof}
Next we prove Proposition~\ref{prop:CCG-graph-theoretically}.
\begin{proof}[Proof of Proposition~\ref{prop:CCG-graph-theoretically}]
Let $\hat H(z,x)$ be defined as the right-hand side of (\ref{eq:H-as-tree-sum}). It is finite and analytic on $\{z:|z|\le 1-\delta\}$ for any $\delta>0$
since $\left|\mathcal{T}_v^{\text {sim }}\right|=\frac{(v-1)!v^{s-1}}{s!((q-1)!)^s}$ with $s=(v-1)/(q-1)$, which by Stirling's approximation is  
$e^{O(v)}v!$. We claim that $\hat H=H$. 
Recall the definition of $B_{H,r}(X,x)$ from (\ref{eq:Bhr}). Suppose $T$ is an $r$-rooted tree with at least one hyperedge. 
Let $e_1,\ldots,e_d$ be the hyperedges containing $r$.
Let $T_{e_i, k}, i\in [d], k\in [q-1]$ be the subtrees generated by the  $k$-th node of hyperedge $e_i$, where we assume that the node $r$ is the last node $q$
in the hyperedge $e_i$. The hyperedge $e_i$ itself is  included into the tree $T_{e_i,k}$.
Note that for any two hyperedges $e,f$ of a tree we have $|e\cap f|\in \{0,1\}$.
Thus $|e\cap f|=1$ iff $e\sim f$.
Then we can expand (\ref{eq:Bhr}) further as 
\begin{align*}
B_{T,r}(x)&=\prod_{i\in [d]}K(x,X(e_i))\prod_{e_i,e_j, 1\leq i<j\leq d}K(X(e_i),X(e_j)) \\
&\times \prod_{i\in [d], k\in [q-1]}\prod_{f\in T_{e_i,k}, (i,k)\in f} K(X(e_i),X(f))\prod_{f_1\sim f_2\in T_{e_i,k},f_1,f_2\neq e_i}K(X(f_1),X(f_2)).
\end{align*}
Here $(i,k)$ denotes the $k$-th node of $e_i$. We can simplify this as
\begin{align*}
B_{T,r}(x)&=\prod_{i\in [d]}K(x,X(e_i))\prod_{e_i,e_j, 1\leq i<j\leq d}K(X(e_i),X(e_j))\prod_{i\in [d], k\in [q-1]}B_{T_{e_i,k}}(X(e_i)).
\end{align*}

Consider the  expectation of $B_{T,r}(x)$  conditioned on   $X(e_i)=y_i, i\in [d]$. Then 
\begin{align*}
\E[B_{T,r}(x)| X(e_i)=y_i, i\in [d]] &=
\prod_{i\in [d]}K(x,y_i)\prod_{e_i,e_j, 1\leq i<j\leq d}K(y_i,y_j)
\prod_{i\in [d], k\in [q-1]}b_{T_{e_i,k}}(y_i)
\end{align*}
Then
\begin{align*}
b_{T,r}(x) = \int_{\Omega^v}\prod_{i\in [d]}K(x,y_i)\prod_{e_i,e_j, 1\leq i<j\leq d}K(y_i,y_j)
\prod_{i\in [d], k\in [q-1]}b_{T_{e_i,k}}(y_i) \nu(dy_i).
\end{align*}
From this we obtain that $\hat H$ is a fixed point of (\ref{eq:H-def}), 
where normalizations $1/d!$ and $1/(q-1)!$ are included since the hyperedges $e_i, i\in [d]$ and nodes $i\in [q-1]$
in each hyperedge $e_i, i\in [d]$ are unordered. Therefore, by Proposition~\ref{prop:existsCCG} $\hat H=H$. 
\end{proof}

Our final step is proving Proposition~\ref{prop:t-star-2}. 

\begin{proof}[Proof of Proposition~\ref{prop:t-star-2}]
Introduce $\hat \cR(z)$ and $\hat \cE(z)$ as follows.
\begin{align*}
\hat\cR(z)&=\sum_{v=1}^\infty \frac{z^v}{v!}\sum_{T\in\cT^{\rm sim}_v} |V(T)|b_{T}.  \\
\hat\cE(z)&=\sum_{v=1}^\infty \frac{z^v}{v!}\sum_{T\in\cT^{\rm sim}_v} |E(T)|b_{T}.
\end{align*}
As above, it is easy to see that $\hat \cR,\hat \cE$ are 
analytic on $\{z: |z|\le 1/e^2\}$. 

\begin{lemma}\label{lemma:RvsHatR}
We claim that $\hat \cR=\cR, \hat \cE=\cE$.
\end{lemma}

\begin{proof}
Fix a tree $T\in\cT_v^{\rm sim}$ with a distinguished root $r$. Let $e_1,\ldots,e_d$ hyperedges containing $r$.
We have
\begin{align*}
B_T(X)=\prod_{i,j\in [d]}K(X(e_i),X(e_j))\prod_{i\in [d], k\in [q-1]}B_{T_{i,k}}(X),
\end{align*}
where the trees $T_{i,k}$ are defined as in the proof of Proposition~\ref{prop:CCG-graph-theoretically}. Then
\begin{align*}
b_T=\int_{\Omega^d} \prod_{i,j\in [d]}K(y_i,y_j)\prod_{i, k\in [q-1]}b_{T_{i,k}}(y_i) \nu(dy_i).
\end{align*}
The proof of $\hat\cR=\cR$ is then completed similarly to the proof of Proposition~\ref{prop:CCG-graph-theoretically} where the multiplier $v$ appears in front of $b_T$
due to $v$ choices of the root $r$.

The proof of $\hat\cE=\cE$ is obtained similarly. Fix a tree $T$ and an edge $e$ in this tree. Consider $q$ subtrees rooted at $q$ nodes of $e$. Fixing
the value $X(e)$ to $x$  we obtain $q$ copies of $H(x,z)$. The multiplier $|E(T)|$ appears due to that many choices of $e$. Dividing by $q!$ is due
to symmetry of $q$ nodes in $e$.  
\end{proof}
We now complete the proof of the proposition. Applying Lemma~\ref{lemma:RvsHatR} and $E(T)=(v-1)/(q-1)$ for every $T\in \cT_v^{\rm sim}$ we have
\begin{align*}
\cR(z)-(q-1)\cE(z)&=z+\sum_{v=1}^\infty \frac{z^v}{v!}\sum_{T\in\cT^{\rm sim}_v} |V(T)|b_{T}-
(q-1)\sum_{v=1}^\infty \frac{z^v}{v!}\sum_{T\in\cT^{\rm sim}_v} (v-1)/(q-1) b_{T} \\
&=z+\sum_{v=1}^\infty \frac{z^v}{v!}\sum_{T\in\cT^{\rm sim}_v} b_{T}.
\end{align*}
Since $\cR$ and $\cE$ are analytic on $|z|\le 1/e^2$, then $\cR(\rho z),\cE(\rho z)$ are analytic on $|z|\le 2$ when $|\beta|\le \beta^*$ and $\beta^*$ is sufficiently small.
Then
\begin{align*}
\cF(z) &=\rho^{-1}\left(\rho z+\sum_{v=1}^\infty \frac{(\rho z)^v}{v!}\sum_{T\in\cT^{\rm sim}_v} b_{T}\right) \\
&=z+\sum_{v=1}^\infty \frac{\rho^{v-1}z^v}{v!}\sum_{T\in\cT^{\rm sim}_v} b_{T} \\
&= z+\sum_{v=1}^\infty z^v \mu_v \\
&=\Phi(z).
\end{align*}
This completes the proof of the proposition.
\end{proof}

\section{Numerical results. Comparison with physics derivations}\label{section:Numerics}
In this section we report on the numerical results of computing the free energy limit (\ref{eq:main-result}) and compare it with the 
known analytical results from~\cite{feng2019spectrum} when $q=2$, and with the
results based on the physics derivation in~\cite{maldacena2016remarks} when $q=4$. When $q=2$ the free energy limit is derived
using the random matrix theory and is given as
\begin{align}
{1\over 2\beta}\int_{-2}^2 \log(\cosh(\beta\lambda\over 2)){\sqrt{4-\lambda^2} \over 2\pi}d\lambda. \label{eq:free-energy-rm}
\end{align}
We computed this expression numerically and we have computed numerically our answer given as  (\ref{eq:main-result}). The results
match within the range $\beta\in [0,3]$ which is the range for which we managed to compute (\ref{eq:main-result}). The results 
are reported in Figure~\ref{fig:q2}. We note an excellent agreement of our result (\ref{eq:main-result}) with (\ref{eq:free-energy-rm}).

\begin{figure}
    \centering
    \includegraphics[width=0.8\linewidth]{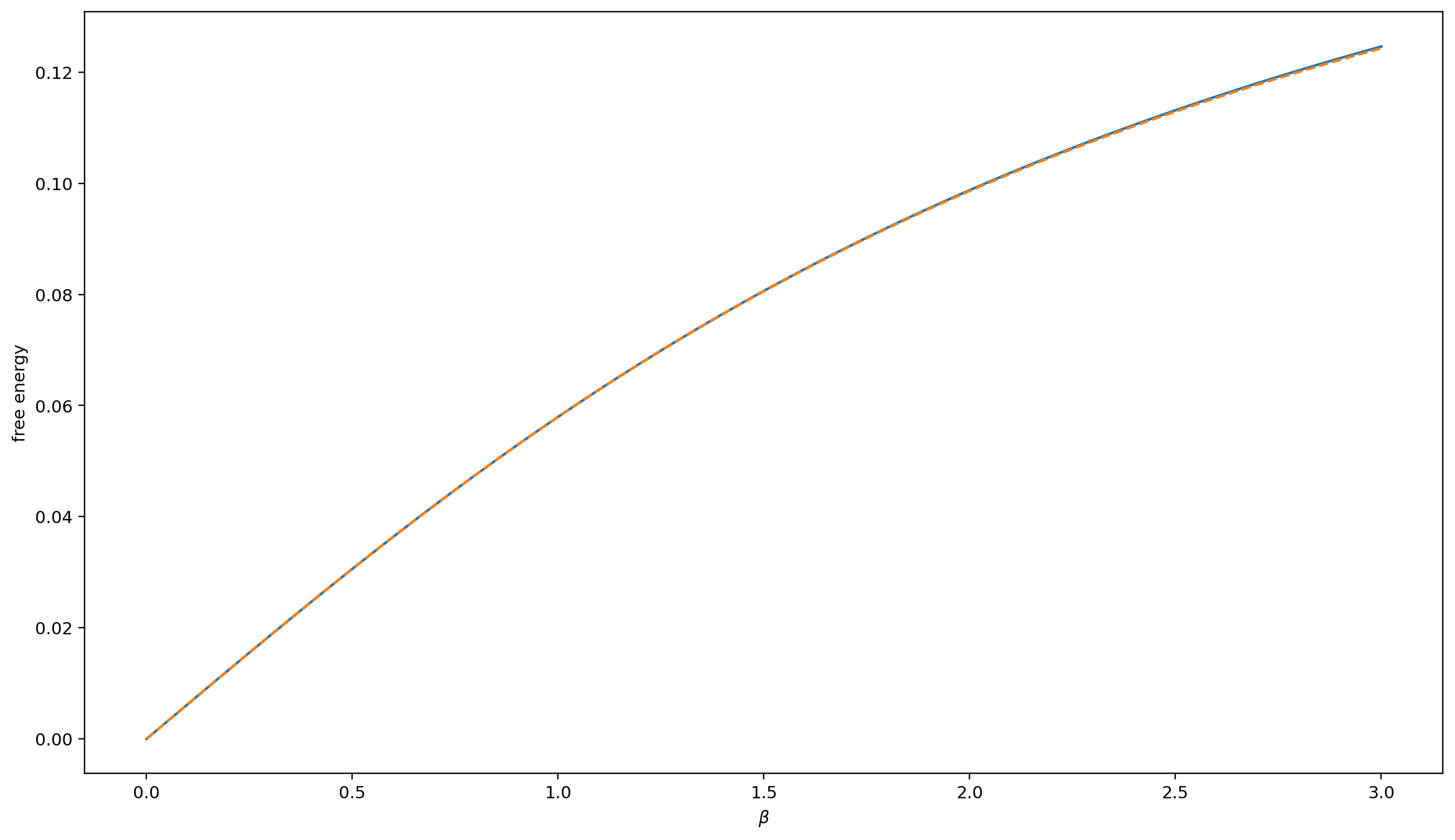}
    \caption{The orange dotted curve is obtained by numerically solving \eqref{eq:main-result} for $q=2$, and the blue curve using (\ref{eq:free-energy-rm}). They agree to within the numerical precision of our solver.}
    \label{fig:q2}
\end{figure}

We begin by describing the answers from this paper first. 
Introduce two functions $G_\beta:[0,\beta]\to\R$ and $\Sigma_\beta:[0,\beta]\to\R$ defined through the following two identities. The first identity is $\Sigma_\beta(\tau)=G_\beta(\tau)^{q-1}$ for all $\tau \in [0,\beta]$. To state the second identity, consider the Fourier transforms $\hat G_\beta(\omega_n), \hat \Sigma_\beta(\omega_n)$ for $\omega_n=(2n+1)\pi/\beta, n=0,1,2,\ldots$. The second identity is
\begin{align}\label{eq:physics-answer}
\hat G_\beta(\omega_n)={1\over -i\omega_n-\hat \Sigma_\beta(\omega_n)}.
\end{align}
The free energy is then predicted as 
\begin{align}
f_{{\rm physics},q}(\beta)\triangleq \frac{1}{\beta}\int_0^\beta \left[\frac{1}{q}\int_0^bG_b(\tau)^qd\tau\right]db. \label{eq:free-energy-physics}
\end{align}
We have verified numerically that $f_{{\rm physics},q}(\beta)=f_q(\beta)$ for $\beta\in [0,3]$ for $q=4$, see Figure~\ref{fig:q4}.

\begin{figure}
    \centering
    \includegraphics[width=0.8\linewidth]{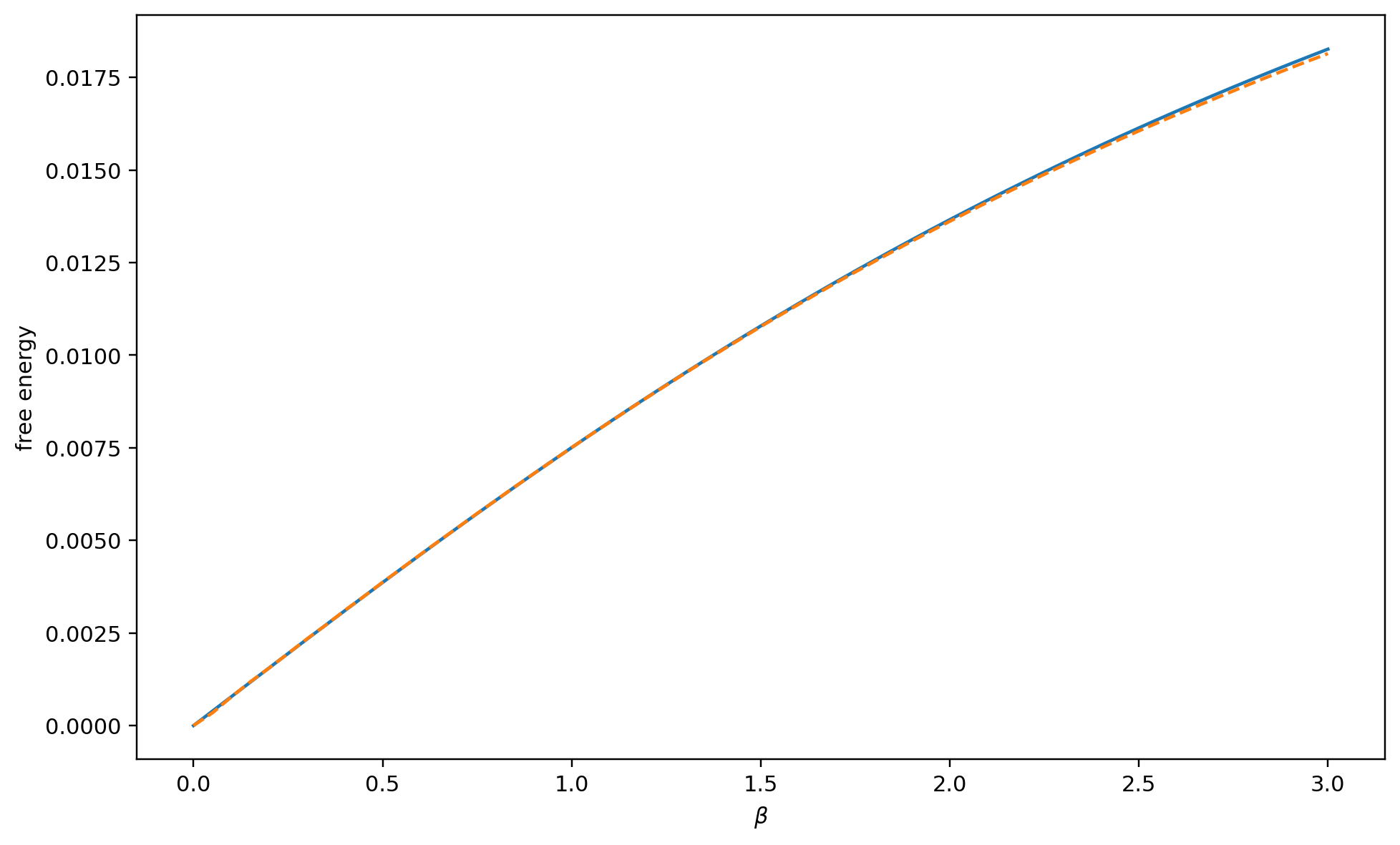}
    \caption{The orange dotted curve is obtained by numerically solving \eqref{eq:main-result} for $q=4$, and the blue curve using (\ref{eq:free-energy-physics}). They agree to within the numerical precision of our solver.}
    \label{fig:q4}
\end{figure}

\section{Open problems}\label{section:open-problems}
A number of questions remain open and we discuss here some of them. First, the agreement between our computation and the computations based on physics methods (path integration plus replica method) are purely numerical. It is quite likely that there is an analytical link between the methods. At the very least we believe that there is a way to reconcile our formula (\ref{eq:main-result}) with the physics-based answer (\ref{eq:physics-answer}). 
More ambitiously, perhaps there is a path to rigorizing the physics answer directly. At first glance this might appear to be beyond reach, since the answer involves the replica trick which remains mathematically unsound. However, we note that the computations are done for the much simpler annealed free energy which is trivialized in the classical domain; hence, there might be a hope of rigorizing the path integration plus replica method for the annealed case.

The next most important question remaining for future research is extending our result to the domain of all inverse temperature $\beta$.
We believe this is achievable by looking at the locally tree-like structure of  hypergraphs induced by the product of local Hamiltonians $H_{I_1}\cdots H_{I_m}$ and establishing some form of correlation decay for the Chord-Cavity-Generator function $H$. Furthermore, we believe it might be possible to compute various observables of interest rigorously using our methods, such OTOCs in real and imaginary times. The scaling behavior of these observables is related to quantities relevant in AdS/CFT, and thus providing a rigorous confirmation of physics-derived answers to these questions in~\cite{maldacena2016remarks} is of great interest. 

Frustratingly, even proving the existence of the free energy limit in all temperature regimes, and similarly the existence of the ground state energy limit remain open. Perhaps these questions alone can be addressed by employing some sort of Gaussian comparison methods as for the classical spin glass case. The complications arise from the fact that states with different levels of entanglement have different variances and this precludes the application of the comparison methods, at least superficially.

Next, we believe our approach is extendable to other models, of which the immediate natural candidate is the quantum spin glass over Paulis. The model is defined in terms of a Hamiltonian $H=\sum_I H_I$ with $I\in {[n]\choose q}$,  $H_I=J_I P_I$, where $J_I$ is centered Gaussian with the same normalization as in the SYK model and $P_I$ ranges over all Paulis supported on $I$. This model was analyzed in~\cite{anschuetz2025bounds} and~\cite{anschuetz2025strongly}. Combining the results of the two papers one can show that the ground state energy is $\Theta(3^{q\over 2}f(q)n)$ with $f(q)\in [{1\over \log^{O(1)}q}, 2\log q]$. We believe that using our method, one should be able to obtain the precise value of the annealed free energy for this model at least for high temperature, but also possibly for all temperature values, and ultimately obtain the asymptotic value of $f(q)$ in the limit $q\to\infty$. We note that in the classical realm quenched and annealed free energies are asymptotically the same as $q$ grows~\cite{GamarnikJagannathKizildag2025Shattering} and we anticipate the same for the quantum counterpart~\cite{swingle2024bosonic}. Finally, we hope our methods can be employed for computing free energies for the SYK and quantum spin glass models on lattices along the lines of the sequential cavity method introduced in~\cite{GamarnikKatzSequentialCavity}. We see no immediate obstructions to the implementation of our method at least for the problem of computing the annealed free energy for lattice models at high temperature.

\section*{Acknowledgements}
The authors gratefully acknowledge enlightening conversations with J.C. Mourrat, Victor Issa and Kuikui Liu.
The work of the first author is funded by ONR grant N000142512545. AS is funded by a Google PhD Fellowship. AZ is funded by a Hertz Fellowship.

\bibliographystyle{amsplain}
\bibliography{bibliography-05.2026}

\end{document}